%% file: main_arxiv.tex
\documentclass[11pt]{article}

\usepackage[a4paper,margin=1in]{geometry}
\usepackage{setspace}
\usepackage{amsmath,amssymb,amsthm,bm}
\usepackage{graphicx}
\usepackage{newtxtext}
\usepackage[subscriptcorrection]{newtxmath}
\usepackage{natbib}
\usepackage[colorlinks=true,linkcolor=blue,citecolor=blue,urlcolor=blue]{hyperref}
\usepackage{orcidlink}

\newtheorem{theorem}{Theorem}[section]
\newtheorem{proposition}[theorem]{Proposition}
\newtheorem{corollary}[theorem]{Corollary}
\newtheorem{lemma}[theorem]{Lemma}
\newtheorem{remark}[theorem]{Remark}

\newcommand{\introductionheading}{\section{Introduction}}

\begin{document}
\onehalfspacing

\title{Fitting Dynamic Conditional Densities\\
to Cross-Sectional Returns with Persistent Tail Intensities}
\author{Yasumasa Matsuda\,\orcidlink{0000-0002-9916-0922} and Keigo Sato\\
Graduate School of Economics and Management\\
Tohoku University, Japan}
\date{}
\maketitle

\begin{center}
\small
Correspondence: Yasumasa Matsuda, Graduate School of Economics and Management,
Tohoku University, Japan.
Email: \href{mailto:yasumasa.matsuda.a4@tohoku.ac.jp}{yasumasa.matsuda.a4@tohoku.ac.jp}
\end{center}

\input{abstract_arxiv}

\begin{center}
\begin{minipage}{0.92\textwidth}
\end{minipage}
\end{center}

\input{sec1}
\input{sec2}
\input{sec3}
\input{sec4}
\input{sec5}
\input{sec6}
\input{backmatter}

\bibliographystyle{apalike}
\bibliography{Bibliography}

\clearpage
\appendix
\renewcommand{\thesection}{S\arabic{section}}
\renewcommand{\theHsection}{supp.\arabic{section}}
\numberwithin{equation}{section}
\setcounter{table}{0}
\renewcommand{\thetable}{S\arabic{table}}
\renewcommand{\theHtable}{supp.\arabic{table}}

\section*{Supplementary material}
This supplement gives the Fisher projection formulas, tail-intensity
proofs, and proofs of consistency and asymptotic normality. It also reports
the simulation coefficients, the M5--M3 forecast comparison, and empirical
data and PCA construction details. References prefixed by ``S'' concern
this supplement. References to the main paper are identified explicitly.

% Generated from the shared supplement with local labels prefixed by S-.
\input{arxiv_supplement}

\end{document}

%% file: abstract_arxiv.tex
\begin{abstract}
We propose dynamic conditional densities with spline mixtures (DCD-SM) to
detect whether the tail intensity of financial returns depends on past
information after accounting for volatility and tail-order dynamics.
The target is a weighted mixture
of the conditional return densities of individual assets in a cross section.
Individual distributions may differ, and returns within a date
may be dependent. A normalized spline-mixture correction uses score-driven
updates to change left and right tail intensities while preserving the
reference tail order.
We use Wald tests for persistence and assess empirical size, power, and
forecast stability in simulations. With fitted references that have dynamic
tail order, Wald tests at the nominal five-percent level detect dynamics
that affect both tail intensities in the same direction in
all 19 US and Japanese equity cross sections, for both raw returns and
residuals after principal-component removal with training-fixed loadings. Evidence for dynamics
in the imbalance between the two tail intensities is more limited.
In the residuals, the
correction improves out-of-sample one-percent lower-tail quantile loss
relative to these references in 16 of 19 panels, with three significant
positive differences at the five-percent level. Forecast gains are most evident in log scores and
one-percent quantile losses and are less uniform for five-percent quantile loss.
\end{abstract}

%% file: sec1.tex
\introductionheading
\label{sec:introduction}

Volatility persistence is a central feature of financial returns.
ARCH and GARCH models describe this persistence through conditional
variance dynamics \citep{Engle1982,Bollerslev1986}.
We ask whether tail characteristics remain
predictable after location and scale dynamics have been taken into account.
We distinguish tail order, which determines the rate of tail decay, from
tail intensity, which scales the tail density for a given order. Our main
question is whether dynamic tail order is sufficient, or whether tail
intensity has additional persistence. We study this question for conditional
return mixtures formed from equity cross sections. With equal weights,
the cumulative distribution function of the mixture gives the probability
that a stock selected at random from the cross section has a return below
a given threshold, conditional on past information.

\citet{Hansen1994} extends ARCH modeling to conditional distributional
parameters beyond the mean and variance. Conditional skewness and
higher-moment persistence are studied by \citet{HarveySiddique1999} and
\citet{JondeauRockinger2003}. Related score-driven models directly describe
degrees-of-freedom dynamics \citep{BlazsekMonteros2017}, time-varying
peaks-over-threshold densities \citep{Massacci2017}, and densities with
asymmetric tail indices \citep{Thiele2020}.
\citet{DInnocenzoLucasSchwaabZhang2024} model score-driven tail shape and
tail scale in a generalized Pareto distribution and document persistence
in both. Their tail-scale parameter can also reflect changes in the scale
of the full return distribution. \citet{Catania2021} develops
dynamic mixtures of return densities with score-driven updates of both
component parameters and mixture weights. More recently,
\citet{vanBrummelenGorgiKoopman2026} use a spline-based error density to
construct score-driven location and scale updates. In those specifications,
the spline density parameters are fixed while location or scale changes
over time.

To capture tail-intensity persistence beyond dynamic location, scale, and
tail order, we propose dynamic conditional densities with spline mixtures
(DCD-SM). We transform returns to $[0,1]$ using a given cumulative
distribution function. We call the distribution used for this
transformation the reference distribution. We represent the conditional
mixture density of the transformed returns by a spline mixture. Its
weights evolve through two score-driven updates
\citep{CrealKoopmanLucas2013}. One changes both tail intensities together,
and the other changes their right--left balance. The spline construction
preserves the reference tail order for any fixed number of components.
We use references with dynamic location and scale, and references that
also have dynamic tail order.
In contrast to \citet{Catania2021}, DCD-SM uses a spline correction on the
reference PIT scale to test additional tail-intensity dynamics while
preserving the tail order of a fitted reference.

The analyst defines the target as a weighted mixture of constituent
conditional return densities. The nonnegative weights sum to one and are
chosen before the current returns are observed. Constituent densities may
differ, and returns within a date may be dependent. The target is the
conditional density of a constituent selected according to these weights.
It generally differs from the distribution of a portfolio return.
\citet{KellyJiang2014} use the cross section of stock
returns to estimate time-varying tail risk.
\citet{KokoszkaMiaoPetersenShang2019} forecast monthly cross-sectional
return densities for the S\&P 500 and Dow Jones using functional time
series methods. \citet{ZhangKokoszkaPetersen2022} develop Wasserstein
autoregressive models for density time series and apply them to the same
cross sections. In their empirical applications, estimated density curves
serve as inputs to the functional time series models. In a related
spline-mixture formulation, \citet{MatsudaIwafuchi2025} develop
density-valued ARMA models through logit transformations of the mixture
weights. DCD-SM fits the conditional mixture density directly to the
observed returns using score-driven updates. The state
update has no additional noise term, so parameter values and initial
conditions determine the state from past returns. The weighted log score
can therefore be evaluated recursively without integrating over unobserved
states or first estimating a density curve for each date. We estimate the
chosen mixture density using a weighted cross-sectional log score. The
conditional expectation of this weighted log score depends only on the
chosen mixture density, even when the constituent conditional densities
differ and returns are dependent within a date.

With the reference distribution treated as known and the spline dimension
fixed, we establish consistency and asymptotic normality of the DCD-SM
estimator under stability and identification conditions. Conditional Wald
tests use a date-level sandwich covariance that allows dependence within
a date. Simulations re-estimate all stages to assess the feasible
procedure's finite-sample size, power, and forecast stability under
independent and common-factor sampling. We separately evaluate strict
out-of-sample (OOS) density and tail-quantile forecasts, since in-sample
detection and predictive gains concern different questions.

The empirical analysis covers 19 US and Japanese equity cross sections,
including the S\&P 500 and Nikkei 225 constituent panels and their sectors.
Conditional on a fitted reference with dynamic tail order, nominal
five-percent Wald tests detect dynamics that affect both tail intensities
in the same direction in all 19 raw-return panels and all 19 panels after removal of a
first principal component with training-fixed means and loadings. Evidence for dynamics in the
right--left imbalance is less widespread. For the principal-component
residuals, the additional correction improves one-percent lower-tail
quantile loss relative to that reference in 16 of 19 panels, with three
significant positive differences at the five-percent level. Forecast gains
are most evident in log scores
and one-percent quantile losses, while five-percent quantile-loss
gains are less uniform.

Our contribution is a framework for directly estimating the dynamics of a
chosen conditional mixture from cross-sectional returns whose conditional
densities may differ and whose dependence need not be specified. Within
this framework, we test additional tail-intensity persistence beyond
dynamic location, scale, and tail order. We use cross-sectional returns
to study density dynamics that can be difficult to detect from a single
return series.

Derivations, proofs, and additional simulation and data-construction
details are provided in the online supplementary material.

%% file: sec2.tex
\section{Dynamic Conditional Densities with Spline Mixtures}
\label{sec:dcd-sm}

We model dynamic conditional densities with spline mixtures (DCD-SM)
relative to a given reference distribution. Observations are
mapped to the probability-integral-transform (PIT) scale of the reference,
where the density correction is represented by a spline mixture with
score-driven dynamics \citep{CrealKoopmanLucas2013}.

For cross-sectional observations, the target is a weighted mixture of
constituent conditional densities, with constituent weights determined
before the current observations arrive. The constituent distributions may
differ, and observations within a date may be conditionally dependent.

\subsection{Spline-mixture density and embedding}
\label{subsec:spline-parameterization}

Let $B_1,\ldots,B_J$ be nonnegative B-spline basis functions on $[0,1]$
forming a partition of unity \citep{deBoor2001},
\[
    \sum_{j=1}^J B_j(u)=1,
\]
and define
\[
    m_j=\int_0^1 B_j(u)\,du,
    \qquad
    \phi_j(u)=\frac{B_j(u)}{m_j},
    \qquad j=1,\ldots,J.
\]
Then each $\phi_j$ integrates to one and $\sum_jm_j=1$. For a probability
vector $p=(p_1,\ldots,p_J)'$ with $p_j>0$ for every $j$, define the spline-mixture density with a
parameter $q\in\mathbb R^{J-1}$ by
\begin{equation}
    h_{\rm SM}(u\mid q)
    =\sum_{j=1}^J p_j(q)\phi_j(u).
    \label{eq:dcd-mixture}
\end{equation}
The parameter $q$ is specified using a Helmert contrast matrix
$Q\in\mathbb R^{(J-1)\times J}$ satisfying
\[
    QQ'=I_{J-1},
    \qquad
    Q\mathbf 1_J=0.
\]
The simplex weights are embedded in Euclidean coordinates by
\begin{equation}
    q=Q\log(p/m),
    \label{eq:logit-embedding}
\end{equation}
where the logarithm and division are componentwise.  The inverse mapping is
the weighted softmax
\begin{equation}
    p_j(q)
    =
    \frac{m_j\exp\{(Q'q)_j\}}
         {\sum_{k=1}^Jm_k\exp\{(Q'q)_k\}},
    \qquad j=1,\ldots,J.
    \label{eq:dcd-weights}
\end{equation}
Consequently,
\[
    q=0
    \quad\Longleftrightarrow\quad
    p=m
    \quad\Longleftrightarrow\quad
    h_{\rm SM}(u\mid q)=1.
\]
Thus the origin in the logit coordinates corresponds exactly to the uniform
density on the PIT scale. We set the static spline anchor at $q_0=0$.
We do not fit a separate static spline correction.

In the empirical analysis we use clamped cubic B-splines
\citep[Chapter~IX, pp.~99--100]{deBoor2001} with breakpoints
\[
    (0,.05,.10,.20,.30,.40,.50,.60,.70,.80,.90,.95,1),
\]
which gives $J=15$ normalized spline-density components.

\subsection{Model description}
\label{subsec:description}

Let $r_{t,i}$ denote the return on stock $i$ at date $t$. The observed
cross section is
\[
    r_{t,1},\ldots,r_{t,N_t}.
\]
Let $\mathcal F_{t-1}$ denote the information available before this cross
section is observed. The constituent
set and its size $N_t\ge1$ are taken to be $\mathcal F_{t-1}$-measurable.
Let $g_{t,i}$ denote the conditional density of $r_{t,i}$ given
$\mathcal F_{t-1}$.  These densities may differ across $i$, and observations
within a date may be conditionally dependent.

The analyst specifies the distributional target by choosing nonnegative,
$\mathcal F_{t-1}$-measurable weights $w_{t,i}$ satisfying
$\sum_{i=1}^{N_t}w_{t,i}=1$. We take the target density to be
\[
    \bar g_t(r):=\sum_{i=1}^{N_t}w_{t,i}g_{t,i}(r).
\]
This is the conditional density of the return on a constituent selected at
random with probabilities $w_{t,i}$, independently of the current returns
given $\mathcal F_{t-1}$. To fit this target, we evaluate each observed
return under the same modeled mixture density. The individual returns need
not have this density and may be dependent. In general, the target differs
from the density of an index or portfolio return.
When $N_t=1$, it reduces to the usual single-series conditional density.
If all constituents have the same conditional density $g_t$, then
$\bar g_t=g_t$ for any such choice of weights. This special case does not
require conditional independence. More generally, the weights define the
mixture of interest to the analyst. They may reflect, for example, market
capitalization or nonnegative portfolio holdings known before date $t$,
and are held fixed with respect to candidate model parameters.
The empirical study in Section~\ref{sec:empirical} uses equal weights,
$w_{t,i}=1/N_t$, so every constituent has the same selection probability.

Let $F_t$ be a given, $\mathcal F_{t-1}$-measurable reference CDF,
continuously differentiable and strictly increasing with positive
density $f_t$ on the return support. Models M1--M3 in
Section~\ref{subsec:m1-m5} provide the specific reference CDFs used below.
Define
\[
    u_{t,i}=F_t(r_{t,i}),
    \qquad i=1,\ldots,N_t.
\]
Write $\phi_{t,i}(u)$ for the conditional density of $u_{t,i}$ given
$\mathcal F_{t-1}$.

DCD-SM represents this target on the PIT scale using the spline mixture
defined in Section~\ref{subsec:spline-parameterization}:
\begin{equation}
    \bar\phi_t(u)
    :=\sum_{i=1}^{N_t}w_{t,i}\phi_{t,i}(u)
    =h_{\rm SM}(u\mid q_t),
    \label{eq:average-pit-density}
\end{equation}
where $q_t$ is $\mathcal F_{t-1}$-measurable. Let
\[
    H_{\rm SM}(u\mid q)
    =\int_0^u h_{\rm SM}(v\mid q)\,dv.
\]
The resulting model CDF for the weighted return mixture is
\begin{equation}
    G_t(r\mid\mathcal F_{t-1})
    =H_{\rm SM}\{F_t(r)\mid q_t\}.
    \label{eq:corrected-cdf}
\end{equation}
Its density is
\begin{equation}
    \bar g_t(r)=g_t(r\mid\mathcal F_{t-1})
    =f_t(r)\,
     h_{\rm SM}\{F_t(r)\mid q_t\}.
    \label{eq:corrected-density}
\end{equation}
Setting $q_t=0$ leaves the reference mixture distribution unchanged exactly.
This identity anchor is what makes the nested comparisons below transparent.

The decomposition into reference and correction is not generally identified
if both are unrestricted. Indeed, writing $H_t(u)=H_{\rm SM}(u\mid q_t)$,
any smooth increasing bijection $A$ of $[0,1]$ gives the same return CDF
through $\widetilde F_t=A\circ F_t$ and
$\widetilde H_t=H_t\circ A^{-1}$. Such a transformed correction need not
belong to the fixed finite-dimensional spline family, but joint
identification of a parameterized reference and correction would still
require separate restrictions. We do not assert that identification here.
We define and fit the correction with the reference held fixed.
Given this reference, linear independence of the spline
basis and the invertible Helmert embedding identify $q_t$ from the corrected
density. Thus the correction and its interpretation are relative to the
chosen reference, rather than an intrinsic decomposition of the return law.

\paragraph{General score-driven dynamics.}
The distributional specification in \eqref{eq:average-pit-density} is
completed by a recursion for the predictable state $q_t$. The basic
DCD-SM specification can be summarized as
\begin{equation}
\begin{aligned}
    u_{t,i}\mid\mathcal F_{t-1}
    &\sim\phi_{t,i}(\cdot),\\
    \sum_{i=1}^{N_t}w_{t,i}\phi_{t,i}(u)
    &=h_{\rm SM}(u\mid q_t),\\
    q_t
    &=\sum_{\ell=1}^{L}K_\ell\,
        s_{q,t-\ell}(q_{t-\ell}).
\end{aligned}
    \label{eq:general-score-dynamics}
\end{equation}
Here $K_\ell\in\mathbb R^{(J-1)\times(J-1)}$, and $s_{q,t}(q)$ is the
derivative with respect to $q$ of the date-$t$ weighted log likelihood
defined below. Only scores from dates before $t$ enter this update.
Equation~\eqref{eq:general-score-dynamics} expresses $q_t$ as a linear
combination of past density scores, following the generalized autoregressive
score (GAS) principle of \citet{CrealKoopmanLucas2013}. The GAS framework
recovers Gaussian ARCH/GARCH dynamics under suitable scaling and parameter
restrictions. DCD-SM extends this score-driven construction from a
conditional variance parameter to a vector of spline-mixture coordinates.

We first give the explicit score function, then introduce its
Fisher-projected version and a parsimonious HAR specification for
tail-intensity dynamics.

\paragraph{Analytic spline score.}
\label{subsec:responsibility-score}
For a candidate state $q$, define the weighted cross-sectional log-score
criterion
\begin{equation}
    \ell_{q,t}(q)
    =\sum_{i=1}^{N_t}w_{t,i}
      \log h_{\rm SM}(u_{t,i}\mid q).
    \label{eq:date-composite-likelihood}
\end{equation}
The state score is $s_{q,t}(q)=\partial\ell_{q,t}(q)/\partial q$.
It has an explicit expression in terms of the component responsibilities
\begin{equation}
    \tau_j(u,q)
    =
    \frac{p_j(q)\phi_j(u)}{h_{\rm SM}(u\mid q)}.
    \label{eq:responsibility}
\end{equation}
Write
\[
    \tau(u,q)
    =\{\tau_1(u,q),\ldots,\tau_J(u,q)\}',
    \qquad
    \bar\tau_t(q)
    =\sum_{i=1}^{N_t}w_{t,i}\tau(u_{t,i},q).
\]
Differentiating analytically gives
\begin{equation}
    s_{q,t}(q)
    =Q\{\bar\tau_t(q)-p(q)\}.
    \label{eq:dcd-score}
\end{equation}
We write $s_{q,t}=s_{q,t}(q_t)$ when the score is evaluated at the current
state. It compares the weighted component responsibilities with the
predicted mixture weights in Helmert coordinates. No numerical
differentiation is required to evaluate this state score.

We use the weighted log score to fit the mixture target chosen by the
analyst. Its algebraic form is that of a
normalized working-independence composite log likelihood
\citep{VarinReidFirth2011}, but it requires neither identical constituent
densities nor within-date independence. Section~\ref{sec:theory}
establishes the score identities under the weighted-average-density
specification in \eqref{eq:average-pit-density}.

\paragraph{Fisher orthogonalization of the spline score.}
\label{subsec:orthogonalization}
To make the location--scale coordinates explicit, write the given reference
CDF as
\[
    F_t(r)=F_t^\circ\!\left(\frac{r-m_t}{\sigma_t}\right),
\]
where $F_t^\circ$ is the standardized reference CDF with density
$f_t^\circ$. The PIT observations and their modeled mixture density are
then
\[
\begin{aligned}
    u_{t,i}
    &:=F_t^\circ\!\left(\frac{r_{t,i}-m_t}{\sigma_t}\right),
      \qquad u_{t,i}\mid\mathcal F_{t-1}\sim\phi_{t,i}(\cdot),\\
    \sum_{i=1}^{N_t}w_{t,i}\phi_{t,i}(u)
    &=h_{\rm SM}(u\mid q_t).
\end{aligned}
\]
The reference density on the return scale is
\[
    f_t(r)=\frac1{\sigma_t}
      f_t^\circ\!\left(\frac{r-m_t}{\sigma_t}\right).
\]
Let
\[
    ms=(m,\log\sigma^2)
\]
denote the location--scale block. Its weighted reference score is
\begin{equation}
    s_{ms,t}
    =\sum_{i=1}^{N_t}w_{t,i}
      \left.
        \frac{\partial}{\partial ms}
        \log\left\{\frac1\sigma
          f_t^\circ\!\left(\frac{r_{t,i}-m}{\sigma}\right)\right\}
      \right|_{ms=(m_t,\log\sigma_t^2)}.
    \label{eq:reference-ms-score}
\end{equation}
The remaining reference coordinates are held fixed in this state
derivative. Thus $s_{ms,t}$ depends on the reference distribution at date
$t$, but not on the spline state $q_t$.

We project the spline score $s_{q,t}(q_t)$ off these location--scale score
directions using fixed reference values, with location--scale block $ms_0$
and spline state $q_0=0$. Their choice is specified below. Let $I_0$ denote
the Fisher information matrix of the return density in
\eqref{eq:corrected-density} with respect to $(ms,q)$, evaluated at these
fixed reference values:
\[
    I_0=
    \begin{pmatrix}
      I_{ms,ms,0} & I_{ms,q,0}\\
      I_{q,ms,0}  & I_{q,q,0}
    \end{pmatrix}.
\]
All Fisher blocks are evaluated once at these fixed reference values and
held fixed over time. Assume that $I_{ms,ms,0}$ is
nonsingular and define
\begin{equation}
    \Pi_{q,ms,0}
    =I_{q,ms,0}I_{ms,ms,0}^{-1}.
    \label{eq:tangent-loading}
\end{equation}
The projected spline score is
\begin{equation}
    s_{q,t}^{\perp}
    =s_{q,t}(q_t)-\Pi_{q,ms,0}s_{ms,t}.
    \label{eq:residual-score}
\end{equation}
In the Fisher metric of the fixed reference, this projection removes the
location--scale score directions from the score used to update $q_t$.

\paragraph{Parsimonious HAR specification.}
\label{subsec:parity-dynamics}
Each unrestricted coefficient matrix $K_\ell$ contains $(J-1)^2$
parameters. To obtain a parsimonious specification, we restrict the spline
state to two coordinates and use the lag structure of the
heterogeneous-autoregressive (HAR) model of \citet{Corsi2009}, with horizons
of 1, 5, and 22 trading days.
Define two vectors in Helmert coordinates by
\begin{equation}
    e_{\rm even}
    =Q
    \begin{pmatrix}
      1&0&\cdots&0&1
    \end{pmatrix}',
    \qquad
    e_{\rm odd}
    =Q
    \begin{pmatrix}
      -1&0&\cdots&0&1
    \end{pmatrix}'.
    \label{eq:endpoint-directions}
\end{equation}
The first and last spline components control the left and right tails
under $u=F_t(r)$. The even vector changes
both tail intensities together, and the odd vector changes their
right--left balance. The sign convention for $e_{\rm odd}$ is right minus
left and matches the empirical implementation. Starting from the
orthogonalized spline score $s_{q,t}^{\perp}$, define the two scalar scores
\begin{equation}
    s_{{\rm even},t}=e_{\rm even}'s_{q,t}^{\perp},
    \qquad
    s_{{\rm odd},t}=e_{\rm odd}'s_{q,t}^{\perp}.
    \label{eq:endpoint-scores}
\end{equation}
For a scalar score sequence $s_t$, define
\begin{equation}
    \bar s_t^{(d)}
    =\frac1d\sum_{\ell=0}^{d-1}s_{t-\ell},
    \qquad d=5,22.
    \label{eq:score-averages}
\end{equation}
The windows of 5 and 22 trading days correspond to approximately one week
and one month, respectively.
The two amplitudes follow HAR$(1,5,22)$ updates:
\begin{equation}
\begin{split}
    a_{{\rm even},t+1}
    &=\kappa_{{\rm even},1}s_{{\rm even},t}
      +\kappa_{{\rm even},5}\bar s_{{\rm even},t}^{(5)}
      +\kappa_{{\rm even},22}\bar s_{{\rm even},t}^{(22)},\\
    a_{{\rm odd},t+1}
    &=\kappa_{{\rm odd},1}s_{{\rm odd},t}
      +\kappa_{{\rm odd},5}\bar s_{{\rm odd},t}^{(5)}
      +\kappa_{{\rm odd},22}\bar s_{{\rm odd},t}^{(22)}.
\end{split}
\label{eq:endpoint-har}
\end{equation}
The spline state is reconstructed as
\begin{equation}
    q_{t+1}
    =a_{{\rm even},t+1}e_{\rm even}
     +a_{{\rm odd},t+1}e_{\rm odd}.
    \label{eq:parity-har-dcd}
\end{equation}
This is a projected-score version of \eqref{eq:general-score-dynamics}
with $L=22$. Each lag matrix is a linear combination of
$e_{\rm even}e_{\rm even}'$ and $e_{\rm odd}e_{\rm odd}'$, with the lag
coefficients determined by the HAR structure.
There are therefore six DCD-SM dynamic coefficients, three for the even
update and three for the odd update. The state is
restricted to $\operatorname{span}\{e_{\rm even},e_{\rm odd}\}$. Its component
in the orthogonal complement is fixed at zero. There is no autoregressive state term,
and the longest HAR window contains scores from the preceding 22 dates.
The odd component allows the corrected density to be asymmetric even when
the reference density is symmetric.

\subsection{Empirical model example: M1--M5}
\label{subsec:m1-m5}

Our aim is to detect persistence in tail behavior beyond location and scale
dynamics. We distinguish tail order, which determines the rate of tail
decay, from tail intensity, which determines the magnitude of the tail
density for a given order. M1--M3 are benchmark models based on existing
methods. M1 is static, M2 adds location and scale dynamics, and M3 adds
tail-order dynamics. We introduce M4 and M5 as DCD-SM extensions of M2 and
M3, respectively. Their spline corrections preserve the reference tail
order and add dynamics in the left and right tail intensities. M5 therefore
tests whether tail-intensity dynamics remain after accounting for dynamic
tail order.
The common parametric anchor is a symmetric Student-$t$ distribution with
location $m$, scale $\sigma>0$, and degrees of freedom $\eta>2$. We use the
variance-standardized parameterization associated with \citet{Hansen1994}:
\begin{equation}
    f_{\rm t}(r\mid m,\sigma,\eta)
    =\frac{\Gamma\{(\eta+1)/2\}}
           {\sigma\sqrt{\pi(\eta-2)}\,\Gamma(\eta/2)}
      \left\{1+\frac{(r-m)^2}{\sigma^2(\eta-2)}\right\}^{-(\eta+1)/2}.
    \label{eq:t-reference-density}
\end{equation}
Thus $m$ and $\sigma^2$ are the mean and variance of this reference density.
Write its CDF as $F_{\rm t}(r\mid m,\sigma,\eta)$ and denote the static
estimates of location, variance, and tail order by
$(m_0,\sigma_0^2,\eta_0)$. In a cross section, these
are parameters of the reference mixture density, rather than common
parameters imposed on each security. M1--M3 retain symmetric references.
M4 and M5 add both even and odd spline corrections, so their final densities
need not be symmetric. The five specifications are defined sequentially
as follows.

\paragraph{M1: static symmetric Student-$t$.}
Student-$t$ distributions have long been used to model stock returns
\citep{Praetz1972,BlattbergGonedes1974}.
The reference CDF is
\[
    F_t(r)=F_{\rm t}(r\mid m_0,\sigma_0,\eta_0),
\]
with $q_t=0$. The parameters $(m_0,\sigma_0^2,\eta_0)$ are estimated by
maximizing the static return log-score criterion subject to
$2.05<\eta_0<100$ and are then held fixed.
This fitted reference supplies the fixed anchor for the Fisher
orthogonalization above, with $ms_0=(m_0,\log\sigma_0^2)$ and $q_0=0$.

\paragraph{M2: dynamic location and scale.}
Starting from M1, keep $\eta_0$ fixed and define
\[
    v_t=\log\sigma_t^2,\qquad
    z_{t,i}=\frac{r_{t,i}-m_t}{\sigma_t}.
\]
Differentiating the log of \eqref{eq:t-reference-density} gives the location and
log-variance scores
\begin{equation}
\begin{aligned}
    s_{m,t}
    &=\sum_{i=1}^{N_t}w_{t,i}
      \frac{(\eta_0+1)z_{t,i}}
           {\sigma_t(\eta_0-2+z_{t,i}^2)},\\
    s_{v,t}
    &=\frac12\sum_{i=1}^{N_t}w_{t,i}
      \left\{\frac{(\eta_0+1)z_{t,i}^2}
                   {\eta_0-2+z_{t,i}^2}-1\right\}.
\end{aligned}
\label{eq:t-ms-explicit-scores}
\end{equation}
The HAR$(1,5,22)$ updates are
\[
    m_{t+1}
    =m_0+\kappa_{m,1}s_{m,t}
       +\kappa_{m,5}\bar s_{m,t}^{(5)}
       +\kappa_{m,22}\bar s_{m,t}^{(22)},
\]
\[
    v_{t+1}
    =\log\sigma_0^2+\kappa_{v,1}s_{v,t}
       +\kappa_{v,5}\bar s_{v,t}^{(5)}
       +\kappa_{v,22}\bar s_{v,t}^{(22)}.
\]
Thus
\[
    F_t(r)
    =F_{\rm t}(r\mid m_t,\sigma_t,\eta_0),
\]
and $q_t=0$.  Setting all six M2 HAR coefficients to zero gives M1 exactly.

Under a Gaussian reference, the score-driven framework includes ARCH models
\citep{CrealKoopmanLucas2013}. M2 belongs to the existing class of
Student-$t$ score-driven models studied by \citet{HarveySucarrat2014}.
Here we use weighted cross-sectional scores and HAR$(1,5,22)$ lags.
These scores account for heavy tails and limit the effect of extreme
returns on the updates.

\paragraph{M3: dynamic Student-$t$ tail order.}
M3 builds on existing Student-$t$ models with score-driven degrees of freedom.
\citet{LucasZhang2016} update variance and degrees of freedom through
score-driven exponentially weighted moving averages, using the inverse
diagonal of the Fisher information for scaling.
\citet{BlazsekMonteros2017} use a score-driven autoregression for degrees
of freedom. \citet{Palumbo2026} models symmetric and asymmetric tail-order
dynamics and develops Lagrange multiplier tests for their presence.

Our M3 specification uses weighted cross-sectional scores to update the
tail order of the chosen conditional mixture density. We project the
tail-order score off location and log variance using Fisher coefficients
evaluated once at M1, and use HAR$(1,5,22)$ averages of the projected scores
in the update. The projection removes these score directions in the
Fisher metric of M1.
We hold the M2 estimates of $m_t$ and $\sigma_t$ fixed for each date.
Tail-order dynamics
are specified in a scalar unconstrained latent coordinate $\xi_{\eta,t}$.
Let
\[
    \Lambda(x)=\frac{1}{1+\exp(-x)},
\]
and map the latent state to the degrees of freedom by
\begin{equation}
    \eta_t
    =\eta_L+(\eta_U-\eta_L)\Lambda(\xi_{\eta,t}).
    \label{eq:t-latent-map}
\end{equation}
The implementation uses
\[
    \eta_L=2.05,\qquad
    \eta_U=100.
\]
We set $\eta_L=2.05$ to ensure finite variance and maintain a numerical
margin above the boundary $\eta=2$.
Hence $2.05<\eta_t<100$ automatically. The M1 tail-order estimate
determines the latent anchor through
\begin{equation}
    \xi_{\eta,0}
    =
    \log\!\left(
      \frac{\eta_0-\eta_L}{\eta_U-\eta_0}
    \right).
    \label{eq:t-latent-anchor}
\end{equation}

The latent coordinate follows an anchored score-HAR$(1,5,22)$ recursion,
\begin{equation}
    \xi_{\eta,t+1}
    =
    \xi_{\eta,0}
    +\kappa_{\eta,1}s_{\xi_\eta,t}^{\perp}
    +\kappa_{\eta,5}\bar s_{\xi_\eta,t}^{\perp,(5)}
    +\kappa_{\eta,22}\bar s_{\xi_\eta,t}^{\perp,(22)}.
\label{eq:t-latent-har}
\end{equation}
The score entering this filter is with respect to the latent coordinate.
If $s_{\eta,t}$ denotes the weighted cross-sectional score with respect to the degrees
of freedom, the chain rule gives
\begin{equation}
\begin{split}
    s_{\xi_\eta,t}
    &=
    \frac{\partial\eta_t}{\partial\xi_{\eta,t}}s_{\eta,t}
    =
    (\eta_U-\eta_L)
    \Lambda(\xi_{\eta,t})
    \{1-\Lambda(\xi_{\eta,t})\}
    s_{\eta,t}.
\end{split}
\label{eq:t-latent-score}
\end{equation}
For M3, the same fixed-anchor projection is applied to the latent
tail-order score after the chain-rule transformation in
\eqref{eq:t-latent-score}. Define
\[
    \dot\eta_0
    =\left.\frac{\partial\eta}{\partial\xi_\eta}\right|_{\xi_{\eta,0}}
    =\frac{(\eta_0-\eta_L)(\eta_U-\eta_0)}{\eta_U-\eta_L}.
\]
The Fisher information is evaluated once at M1 and held fixed. Both scores
in the following projection are evaluated at the current M3 state
$(m_t,\sigma_t,\eta_t)$:
\begin{equation}
\begin{split}
    s_{\xi_\eta,t}^{\perp}
    &=s_{\xi_\eta,t}
      -I_{\xi_\eta,ms,0}I_{ms,ms,0}^{-1}s_{ms,t}\\
    &=s_{\xi_\eta,t}
      -\frac{6\dot\eta_0}{\eta_0(\eta_0-2)(\eta_0+1)}s_{v,t}.
\end{split}
\label{eq:t-explicit-projection}
\end{equation}
Here $s_{v,t}$ is the log-variance component of $s_{ms,t}$. The location
coefficient is zero by symmetry. Section~\ref{S-supp:fisher} gives the closed-form Fisher blocks
and their derivation.
The tail-order state in \eqref{eq:t-latent-map} then defines
\[
    F_t(r)
    =F_{\rm t}(r\mid m_t,\sigma_t,\eta_t),
\]
again with $q_t=0$. Setting the three coefficients in
\eqref{eq:t-latent-har} to zero fixes $\xi_{\eta,t}=\xi_{\eta,0}$ and
therefore $\eta_t=\eta_0$, giving M2 exactly.

\paragraph{M4: M2 reference plus DCD-SM.}
We hold the M2 estimates of $m_t$ and $\sigma_t$ fixed for each date and
form the transformed observations as
\[
    u_{t,i}=F_{\rm t}
    (r_{t,i}\mid m_t,\sigma_t,\eta_0).
\]
For M4, \eqref{eq:residual-score} projects the spline score off location
and log variance. At the symmetric M1 anchor, the symmetric knot sequence
makes the even signal orthogonal to location and the odd signal orthogonal
to log variance. Thus \eqref{eq:endpoint-scores} becomes
\begin{equation}
\begin{aligned}
    s_{{\rm even},t}
    &=e_{\rm even}'s_{q,t}-c_{{\rm even},0}s_{v,t},\\
    s_{{\rm odd},t}
    &=e_{\rm odd}'s_{q,t}-c_{{\rm odd},0}s_{m,t},
\end{aligned}
\label{eq:explicit-endpoint-projection}
\end{equation}
where
\begin{equation}
\begin{aligned}
    c_{{\rm even},0}
    &=\frac{2(\eta_0+3)}{\eta_0}\,e_{\rm even}'I_{q,v,0},\\
    c_{{\rm odd},0}
    &=\frac{\sigma_0^2(\eta_0-2)(\eta_0+3)}{\eta_0(\eta_0+1)}
      \,e_{\rm odd}'I_{q,m,0}.
\end{aligned}
\label{eq:endpoint-projection-coefficients}
\end{equation}
The vectors $I_{q,m,0}$ and $I_{q,v,0}$ are the columns of $I_{q,ms,0}$.
Their expectations involve the spline basis and are evaluated by numerical
quadrature once at M1, as detailed in Section~\ref{S-supp:fisher}.
The resulting even and odd scores are propagated by
\eqref{eq:endpoint-har}--\eqref{eq:parity-har-dcd}.

\paragraph{M5: M3 reference plus DCD-SM.}
We hold the M3 estimates of $m_t$, $\sigma_t$, and $\eta_t$ fixed for each
date and use $F_{\rm t}(\cdot\mid m_t,\sigma_t,\eta_t)$ as the reference
CDF. The same DCD-SM recursion is then fitted. The
spline score is projected off location and log variance only. The dynamic
tail-order direction associated with $\xi_{\eta,t}$ is deliberately
retained. M5 therefore combines dynamic Student-$t$ tail order with an
additional even/odd tail-intensity correction.
The even and odd scores use \eqref{eq:explicit-endpoint-projection} with the
same M1 coefficients as M4, evaluating the PIT observations and reference
scores using the M3 estimates for each date.

Within this sequential specification, setting the newly introduced dynamic
coefficients to zero gives the exact zero-restriction relations
\[
    M1\subset M2\subset M3\subset M5,
    \qquad
    M2\subset M4.
\]
M3 and M4 are two different ways of adding shape dynamics beyond the
location and scale dynamics estimated by M2. M5 asks whether tail-intensity dynamics remain useful
conditional on the dynamic Student-$t$ reference. Its spline score retains
the latent tail-order direction, so M5 does not test a Fisher-orthogonal
complement to M3.

\subsection{Estimation, inference, and forecasting}
\label{subsec:estimation-forecasting}

All models are fitted directly to the return observations, using either
raw returns or PC1 residuals in the empirical analysis, without separate
RM/RV standardization. M1 is estimated first. M2 is then
estimated from score-HAR updates of location and log variance. For M3, we
hold the M2 estimates of $m_t$ and $\sigma_t$ fixed for each date and
estimate the three HAR coefficients for $\xi_{\eta,t}$. For M4, we retain
the same M2 estimates and estimate the six DCD-SM coefficients. For M5, we
hold the M3 estimates of $m_t$, $\sigma_t$, and $\eta_t$ fixed for each
date and estimate the six DCD-SM coefficients.
Section~\ref{sec:theory} establishes consistency and asymptotic normality
of the DCD-SM coefficient estimators with the reference distribution and
the Fisher projection matrices treated as known.

For M4 and M5, let
\[
    \kappa
    =(
      \kappa_{{\rm even},1},\kappa_{{\rm even},5},\kappa_{{\rm even},22},
      \kappa_{{\rm odd},1},\kappa_{{\rm odd},5},\kappa_{{\rm odd},22}
      )'.
\]
With the reference distribution at each date and the M1 Fisher projection
matrices held fixed,
the objective is the time average of the weighted cross-sectional return log score
\begin{equation}
    L_T(\kappa)
    =\frac1T\sum_{t=1}^T
      \sum_{i=1}^{N_t}w_{t,i}
      \log\bar g_t(r_{t,i}\mid\kappa),
    \label{eq:dcd-objective}
\end{equation}
where $\bar g_t$ is given by \eqref{eq:corrected-density}. The criterion
uses the same constituent weights that define the target mixture. Each return is
evaluated under the target mixture density $\bar g_t$. For conditionally
iid returns within each date and equal weights, the date-$t$ contribution
is the joint log likelihood divided by $N_t$. In general, the individual
returns need not have this density and may be dependent. The objective
fits the target mixture density without specifying the joint distribution
of the observed cross section. It has the same algebraic form as a
normalized working-independence composite log likelihood
\citep{Lindsay1988,VarinReidFirth2011}.

The reference distribution at each date and the Fisher projection matrices
are held fixed during estimation and covariance calculation. For every candidate
$\kappa$, the state, responsibilities, projected score, even and odd scores, and
next-period state are recomputed forward. We estimate $\kappa$ by directly
maximizing \eqref{eq:dcd-objective}, without an EM step. Inference uses Wald
tests with conditional sandwich covariance estimates.

Strict out-of-sample forecasts are generated recursively. All fitted coefficients
and Fisher projection matrices are held fixed at their training estimates, and the
scores from the last 22 training dates initialize the HAR histories. The density
forecast for date $t$ uses only information and scores dated $t-1$ or earlier.
After the date-$t$ return or cross section is observed, its score is used only
for subsequent forecasts. In a cross section, forecast densities and their
quantiles describe the target mixture distribution, rather than each
constituent's conditional distribution. The same procedure applies
when $N_t=1$, although the numerical comparisons in this paper concentrate
on multiple-observation cross sections.

\subsection{Tail intensities as \texorpdfstring{$J\to\infty$}{J tends to infinity}}
\label{subsec:endpoint-interpretation}

We study how the even and odd amplitudes relate to tail intensity as the
number of spline components $J$ increases. Let $B_{j,J}$ be a clamped
B-spline basis of fixed degree, with masses
$m_{j,J}=\int_0^1 B_{j,J}(u)\,du$, and write $h_J$ for its correction
density. The vectors in \eqref{eq:endpoint-directions} now
depend on $J$, and the state is
\[
 q_{t,J}=a_{{\rm even},t}e_{{\rm even},J}
          +a_{{\rm odd},t}e_{{\rm odd},J}.
\]
We hold the amplitude sequences fixed as $J$ increases. Define
\[
 \alpha_{L,t}=a_{{\rm even},t}-a_{{\rm odd},t},\qquad
 \alpha_{R,t}=a_{{\rm even},t}+a_{{\rm odd},t},
\]
and write the normalizing factor as
\begin{equation}
 Z_{t,J}=1+m_{1,J}(e^{\alpha_{L,t}}-1)
           +m_{J,J}(e^{\alpha_{R,t}}-1).
 \label{eq:endpoint-normalizer}
\end{equation}
For each $J$, the clamped basis gives
\begin{equation}
 h_J(0\mid q_{t,J})=\frac{e^{\alpha_{L,t}}}{Z_{t,J}},\qquad
 h_J(1\mid q_{t,J})=\frac{e^{\alpha_{R,t}}}{Z_{t,J}}.
 \label{eq:endpoint-density-values}
\end{equation}
\begin{proposition}
\label{thm:endpoint-intensity}
Suppose that the support widths of $B_{1,J}$ and $B_{J,J}$ tend to zero
as $J\to\infty$, so that
$m_{1,J}\to0$ and $m_{J,J}\to0$, and that
\[
 \sup_t\bigl(|a_{{\rm even},t}|+|a_{{\rm odd},t}|\bigr)<\infty.
\]
Then $\sup_t|Z_{t,J}-1|\to0$ and
\begin{equation}
\begin{aligned}
 h_J(0\mid q_{t,J})
   &=e^{a_{{\rm even},t}-a_{{\rm odd},t}}+o(1),\\
 h_J(1\mid q_{t,J})
   &=e^{a_{{\rm even},t}+a_{{\rm odd},t}}+o(1),
\end{aligned}
 \label{eq:endpoint-intensity-limit}
\end{equation}
uniformly in $t$. For a reference density $f_t>0$ on $\mathbb R$, the
corresponding return density
$\bar g_{t,J}(r)=f_t(r)h_J\{F_t(r)\mid q_{t,J}\}$ satisfies
\begin{equation}
\begin{aligned}
 \lim_{J\to\infty}\lim_{r\to-\infty}
    \frac{\bar g_{t,J}(r)}{f_t(r)}
   &=e^{a_{{\rm even},t}-a_{{\rm odd},t}},\\
 \lim_{J\to\infty}\lim_{r\to+\infty}
    \frac{\bar g_{t,J}(r)}{f_t(r)}
   &=e^{a_{{\rm even},t}+a_{{\rm odd},t}}.
\end{aligned}
 \label{eq:return-tail-intensity-limit}
\end{equation}
The tail limit is taken for each fixed $J$ before letting $J$ increase.
\end{proposition}

This gives a semiparametric interpretation of the correction. The
Student-$t$ reference determines tail order, and the spline correction
adjusts the left and right tail intensities. The even amplitude controls
their common change in log intensity, and the odd amplitude controls the
difference between the right and left log intensities.
Under the same conditions, the integrated absolute difference satisfies
\[
 \sup_t\int_0^1|h_J(u\mid q_{t,J})-1|\,du
     =O(m_{1,J}+m_{J,J})\to0.
\]
The tail intensity ratios can therefore stay different from one while
this integral tends to zero. Section~\ref{S-supp:endpoints} proves both results.

%% file: sec3.tex
\section{Asymptotic theory}
\label{sec:theory}

The analyst chooses nonnegative weights $w_{t,i}$ summing to one, with
zero weights allowed, to define the conditional mixture of interest. We
study the dynamics of this chosen mixture while allowing constituent
densities to differ and returns within a date to be dependent.
We treat the reference distribution and the Fisher projection matrices as
known and study estimation of $h_{\rm SM}(\cdot\mid q_t)$ with a fixed
spline dimension $J$. We establish consistency and asymptotic normality of
the estimators of its dynamic coefficients as the number of dates $T$
increases. Detailed assumptions and proofs appear in
Section~\ref{S-supp:assumptions} of the supplement.

\subsection{Estimation criterion and regularity conditions}
\label{subsec:oracle-criterion}

Let $E=(e_{\rm even},e_{\rm odd})$ and
$z_t=(a_{{\rm even},t},a_{{\rm odd},t})'$, so that $q_t=Ez_t$.
For a candidate vector $\theta$ of the six HAR coefficients, define
\[
 s_t(q)=\sum_i w_{t,i}\psi(u_{t,i},q),\qquad
 \psi(u,q)=Q\{\tau(u,q)-p(q)\},\qquad u_{t,i}=F_t(r_{t,i}).
\]
The projected signal driving the state recursion is
\[
 d_t(\theta)=E'\{s_t(q_t(\theta))+b_t\},\qquad
 b_t=-I_{q,ms,0}I_{ms,ms,0}^{-1}s_{ms,t}.
\]
The projection coefficient in $b_t$ is evaluated once at the fixed reference
values $(ms,q)=(ms_0,0)$, as in \eqref{eq:tangent-loading}. The reference
score $s_{ms,t}$ in \eqref{eq:reference-ms-score} is computed from the
date-$t$ returns using the known reference distribution. With
$C_w(\theta)=\operatorname{diag}(\kappa_{{\rm even},w},
\kappa_{{\rm odd},w})$, the score update equation is
\begin{equation}
 z_{t+1}(\theta)=\sum_{w\in\{1,5,22\}}C_w(\theta)
                     \frac1w\sum_{j=0}^{w-1}d_{t-j}(\theta).
 \label{eq:oracle-endpoint-recursion}
\end{equation}
This is a stochastic recurrence in the two amplitudes and their 21 lags.
Stability depends on the derivative of the score feedback as well as on
the HAR coefficients.

The local criterion and estimator are
\begin{equation}
 \ell_t(\theta)=\sum_i w_{t,i}
      \log h_{\rm SM}\{u_{t,i}\mid q_t(\theta)\},\qquad
 \widehat\theta_T\in\arg\max_{\theta\in\mathcal N}
                     \frac1T\sum_{t=1}^T\ell_t(\theta),
 \label{eq:local-estimator}
\end{equation}
where $\mathcal N$ is a compact convex neighborhood of the true parameter.
The fixed reference log density is constant in $\theta$, so this criterion
has the same maximizer as the weighted cross-sectional return log score.
Let $\theta_0$ denote the true coefficient vector. Define $S_t$ and the
negative expected Hessian $\mathcal H$ by
\[
 S_t:=\frac{\partial\ell_t(\theta_0)}{\partial\theta},\qquad
 \mathcal H
 :=-E\left[
   \frac{\partial^2\ell_t(\theta_0)}
        {\partial\theta\,\partial\theta'}
 \right].
\]

The assumptions are summarized below. Their full versions appear in the
supplement.
\begin{enumerate}
 \item\label{ass:oracle-specification}
 The reference CDF, constituent set, $N_t$, and weights $w_{t,i}$ are given
 before the date-$t$ returns are observed and do not depend on $\theta$.
 The weighted average conditional PIT density equals
 $h_{\rm SM}\{\cdot\mid q_t(\theta_0)\}$.
 \item\label{ass:filter-stability}
 The true parameter is interior to $\mathcal N$. For $\theta\in\mathcal N$,
 $q_t(\theta)$ is continuous in $\theta$ and jointly stationary and ergodic.
 The effect of
 initialization decays exponentially, uniformly on $\mathcal N$.
 There is a finite constant $C$ such that $\|b_t\|\le C$ for all $t$
 almost surely.
 \item\label{ass:dynamic-identification}
 If $q_t(\theta)=q_t(\theta_0)$ almost surely for every $t$, then
 $\theta=\theta_0$.
 \item\label{ass:derivative-regularity}
 The first and second derivatives of $q_t(\theta)$ with respect to
 $\theta$ are continuous in $\theta$. Together with $q_t(\theta)$, they
 are jointly stationary and ergodic. Their initialization
 effects decay exponentially, uniformly on $\mathcal N$. For some
 $\delta>0$, assume
 \[
 \begin{aligned}
 E\left[\sup_{\theta\in\mathcal N}
   \left\|\frac{\partial q_t(\theta)}{\partial\theta'}\right\|_F^{2+\delta}
   \right]&<\infty,\\
 E\left[\sup_{\theta\in\mathcal N}
   \left\|\frac{\partial^2 q_{t,j}(\theta)}
     {\partial\theta\,\partial\theta'}\right\|_F^{1+\delta}
   \right]&<\infty,\qquad j=1,\ldots,J-1.
 \end{aligned}
 \]
 Here $q_{t,j}$ is the $j$th coordinate of $q_t$, and $\|\cdot\|_F$
 denotes the Frobenius norm.
 \item\label{ass:information-variance}
 The negative expected Hessian $\mathcal H$ is positive definite. There
 is a finite nonrandom matrix $\Omega$ such that
 \begin{equation}
 \frac1T\sum_{t=1}^T
   \operatorname{Var}(S_t\mid\mathcal F_{t-1})
   \xrightarrow{p}\Omega.
 \label{eq:predictable-variance-limit}
 \end{equation}
\end{enumerate}
These conditions concern the score update equation and aggregate
conditional moments.
They do not impose stationarity or independence on the full observation
process. Local stability and identification are assumptions, rather than
consequences of optimizer convergence.

\subsection{Consistency and asymptotic normality}
\label{subsec:consistency}

\begin{theorem}
\label{thm:consistency}
Under Assumptions~\ref{ass:oracle-specification}--\ref{ass:dynamic-identification},
the local estimator in \eqref{eq:local-estimator} converges almost surely to
$\theta_0$.
\end{theorem}

Let $L(\theta):=E[\ell_t(\theta)]$. Under
Assumption~\ref{ass:oracle-specification},
\[
\begin{aligned}
 E[\ell_t(\theta)\mid\mathcal F_{t-1}]
 &=\sum_i w_{t,i}\int_0^1\phi_{t,i}(u)
       \log h_{\rm SM}(u\mid q_t(\theta))\,du\\
 &=\int_0^1 h_{\rm SM}(u\mid q_t(\theta_0))
       \log h_{\rm SM}(u\mid q_t(\theta))\,du.
\end{aligned}
\]
This equality relies on using the same weights in the criterion and in
the target mixture. The expected log-score gap is therefore
\[
 L(\theta_0)-L(\theta)=
 E\left[\operatorname{KL}\left\{
 h_{\rm SM}(\cdot\mid q_t(\theta_0))\,
 \Vert\,h_{\rm SM}(\cdot\mid q_t(\theta))\right\}\right]\ge0.
\]
Linear independence of the spline basis and the invertible logit embedding
identify the density state for the given reference. Dynamic identification then identifies the HAR
coefficients. The full argument is given in the supplement.

Let $J_t=\partial q_t(\theta_0)/\partial\theta'$. The chain rule gives
\begin{equation}
 S_t=J_t's_t(q_t(\theta_0)).
 \label{eq:parameter-score}
\end{equation}
Similarly, Assumption~\ref{ass:oracle-specification} gives
\[
\begin{aligned}
 E[s_t(q_t(\theta_0))\mid\mathcal F_{t-1}]
 &=\sum_i w_{t,i}\int_0^1
       \psi(u,q_t(\theta_0))\phi_{t,i}(u)\,du\\
 &=\int_0^1\psi(u,q_t(\theta_0))
       h_{\rm SM}(u\mid q_t(\theta_0))\,du\\
 &=0.
\end{aligned}
\]
Hence $E(S_t\mid\mathcal F_{t-1})=0$. Thus $S_t$ is a martingale
difference even when individual constituent scores have nonzero
conditional means.

Under Assumption~\ref{ass:oracle-specification}, the negative expected Hessian
has the following representation:
\begin{align}
 \mathcal I_q(q)&=\int_0^1\psi(u,q)\psi(u,q)'
                               h_{\rm SM}(u\mid q)\,du,\nonumber\\
 \mathcal H&=E[J_t'\mathcal I_q(q_t(\theta_0))J_t],\qquad
 \Sigma_t=\operatorname{Var}\{s_t(q_t(\theta_0))\mid\mathcal F_{t-1}\},
 \label{eq:sensitivity-matrix}\\
 \operatorname{Var}(S_t\mid\mathcal F_{t-1})&=J_t'\Sigma_tJ_t.\nonumber
\end{align}

For an ordinary correctly specified log likelihood, the negative expected
Hessian equals the score variance. Here the weighted log score is used
without specifying the joint distribution within a date. Therefore
$\Omega$ cannot in general be recovered from $\mathcal H$ alone, and the
limiting covariance takes the sandwich form used in composite likelihood
inference \citep{VarinReidFirth2011}.

\begin{theorem}
\label{thm:asymptotic-normality}
Under Assumptions~\ref{ass:oracle-specification}--\ref{ass:information-variance},
the local estimator satisfies
\begin{equation}
 \sqrt T(\widehat\theta_T-\theta_0)
 \Longrightarrow N(0,\mathcal H^{-1}\Omega\mathcal H^{-1}).
 \label{eq:asymptotic-normality}
\end{equation}
The conclusion also holds when the states and their derivatives are
computed from initial values satisfying the forgetting conditions.
\end{theorem}

With equal weights and a constant $N$, conditional iid sampling gives
$\Omega=\mathcal H/N$, so the covariance in
\eqref{eq:asymptotic-normality} reduces to $\mathcal H^{-1}/N$.
This simplification is not generally available with heterogeneous or
dependent returns. The proof combines a martingale central limit theorem for $S_t$, a uniform
law of large numbers for the recursive Hessian, and consistency. These
results concern the DCD-SM dynamic coefficients with an M2 or M3 reference
treated as known.

\subsection{Sandwich covariance and conditional Wald tests}
\label{subsec:sandwich-inference}

Using the score update equation with fixed initial values, let
$\widehat S_t=\partial_\theta\widehat\ell_t(\widehat\theta_T)$ and define
\begin{equation}
 H_T=-\sum_t\partial_\theta^2\widehat\ell_t(\widehat\theta_T),\qquad
 B_T=\sum_t(\widehat S_t-\bar S)(\widehat S_t-\bar S)',\qquad
 \widehat V_T=H_T^{-1}B_TH_T^{-\prime}.
 \label{eq:estimated-sandwich}
\end{equation}
Under the assumptions of Theorem~\ref{thm:asymptotic-normality},
$T\widehat V_T\to\mathcal H^{-1}\Omega\mathcal H^{-1}$ in probability.
For a full-row-rank restriction matrix $R$, with
$R\theta_0=r$ and positive-definite restricted limiting covariance,
\[
 W_T=(R\widehat\theta_T-r)'
       (R\widehat V_TR')^{-1}(R\widehat\theta_T-r)
       \Longrightarrow\chi^2_{\operatorname{rank}(R)}.
\]
The even and odd block restrictions each have rank three. Their joint
restriction has rank six.

Lag-zero outer products are sufficient because the parameter score
is a martingale difference under the stated assumptions. A fixed-lag HAC
estimator has the same limit. The main implementation uses lag zero.
This parameter covariance is separate from the
HAC(22) variance used for OOS forecast-score differences.

Under the stated assumptions, we can therefore estimate and test the
dynamics of the chosen mixture without specifying each constituent density
or the joint distribution within a date. Sections~\ref{sec:simulation} and
\ref{sec:empirical} study the numerical performance of the procedure with
estimated reference distributions.

%% file: sec4.tex
\section{Simulation evidence}
\label{sec:simulation}

All designs in this section use $w_{t,i}=1/N$, the equal-weight special
case of the weighted-mixture framework in Section~\ref{subsec:description}.
Section~\ref{sec:theory} treats the reference distribution and Fisher
projection matrices as known. Here we fit M1--M5 and compute the Fisher
projection coefficients from the M1 estimates in each replication to assess
the finite-sample size and power of the conditional Wald tests. The
experiment uses a symmetric Student-$t$ reference. The main comparison is between
conditionally independent draws and draws with a common factor explaining
10 percent of individual return variance on average. We use
$N\in\{20,50,100\}$ and 100 Monte Carlo replications for each design.
We evaluate empirical size under the null of no shape persistence and
power against alternatives with persistent tail dynamics. As an additional
assessment, we examine the
forecast consequences of fitting shape dynamics through out-of-sample
(OOS) accuracy.

\subsection{Generating dynamics and sequential estimation}
\label{subsec:simulation-design}

The generating baseline is the static symmetric Student-$t$ M1, with
location zero, scale one, and degrees of freedom $\eta_0=2.6750886053$.
Returns are expressed in percent.
Since location and scale have no generating dynamics, M2 coincides with
M1 and is omitted as a separate generating model. M2 is still fitted in
estimation.
These values determine generation and are not supplied to estimation.
We generate M3 by adding tail-order dynamics through
\eqref{eq:simulation-eta}, with no spline correction ($q_t=0$).
M4 keeps tail order fixed and generates the even/odd amplitudes
through \eqref{eq:simulation-endpoint}. M5 combines the tail-order
update in \eqref{eq:simulation-eta} with the amplitude updates in
\eqref{eq:simulation-endpoint}.
The conditional generating CDF is
\[
    G_t(y)=H_{\rm SM}\{F_t(y)\mid q_t\},
    \qquad
    q_t=e_{\rm even}a_{{\rm even},t}
        +e_{\rm odd}a_{{\rm odd},t}.
\]
Here $F_t$ is a symmetric Student-$t$ CDF with the fixed generating location
and scale and either static or dynamic tail order, according to the design.
We use the spline basis and the vectors $e_{\rm even}$ and $e_{\rm odd}$
defined in Section~\ref{sec:dcd-sm}.

We use the latent tail coordinate $\xi_{\eta,t}$ in
\eqref{eq:t-latent-map}. Let $d_{\eta,t}$ be its date-average score,
projected off location and log variance at the fixed baseline Fisher
geometry, and let $d_{b,t}$ denote the corresponding projected score for
$b\in\{{\rm even},{\rm odd}\}$. The generating updates are
\begin{align}
    \xi_{\eta,t+1}
    &=\xi_{\eta,0}+
       \sum_{w\in\{1,5,22\}}\kappa_{\eta,w}(N)\bar d_{\eta,t}^{(w)},
       \label{eq:simulation-eta}\\
    a_{b,t+1}
    &=\sum_{w\in\{1,5,22\}}\kappa_{b,w}(N)\bar d_{b,t}^{(w)}.
       \label{eq:simulation-endpoint}
\end{align}
Each date's signals are calculated from the actual $N$ generated returns.
For generation, the Fisher projection coefficients are computed once at
$q=0$ and the true static Student-$t$ reference values, and held fixed.
Only location and log variance are projected out. The tail-order direction
is retained in the spline signal, and no score whitening is applied.

For active components, the adopted S\&P 500 HAR vectors are multiplied by
\[
    \kappa_b(N)=\sqrt{N/500}\,\kappa_b^{\rm SP},
    \qquad b\in\{\eta,{\rm even},{\rm odd}\}.
\]
The baseline HAR coefficient vectors $\kappa_b^{\rm SP}$ are reported in
Section~\ref{S-supp:simulation-calibration} of the supplement. We use the
same coefficients for independent draws and the common-factor designs.
Under independence, the $N^{-1/2}$ fluctuation scale of an average score
motivates this adjustment. The same multiplier is used in the common-factor
design, without an additional dependence-specific variance adjustment.
Consequently, comparisons across $N$ concern this scaled design, not a
fixed-coefficient experiment isolating only sample size. Comparisons between
dependence specifications at a given $N$ use the same coefficients.

We consider five designs. Case 0 uses the static M1=M2 baseline.
Case 1 generates M3 with persistent $\eta$. Case 2 generates M4 with
only the even amplitude persistent. Case 3 generates M5 with persistent
$\eta$ and even amplitude. Case 4 adds persistence in the M5 odd amplitude.
Inactive coefficient vectors are
zero. We generate 500 burn-in dates, 1,886 estimation dates, and 642 OOS
dates. For each candidate coefficient vector, the score update equations
are run through the burn-in observations. These dates are excluded from
the estimation criterion and parameter-score outer products.

Estimation begins with a static symmetric-$t$ M1 fit of location, scale,
and tail order. For fitting, the Fisher projection coefficients are
computed once at $q=0$ and the fitted M1 reference values, and held fixed.
M2 estimates six HAR coefficients for location and log variance, holding
the M1 estimates fixed. M3 holds the M2 estimates of $m_t$ and $\sigma_t$
fixed for each date and estimates three HAR coefficients for
$\xi_{\eta,t}$. M4 uses the same M2 estimates and estimates six even/odd
DCD-SM coefficients. M5 holds the M3 estimates of $m_t$, $\sigma_t$, and
$\eta_t$ fixed for each date and estimates six DCD-SM coefficients.
Although the generating location and scale are static, M2 allows dynamics
in both coordinates.
For estimation, returns are converted from percent to decimal units.
Forecast scores are converted back to the original percentage units.

\subsection{Dependence with identical conditional marginals}
\label{subsec:simulation-dependence}

Returns are generated in two steps. Given a uniform variable $V_{t,i}$,
first draw a PIT value from the spline-mixture density and then apply
the inverse reference CDF to return to the return scale:
\[
    u_{t,i}=H_{\rm SM}^{-1}(V_{t,i}\mid q_t),
    \qquad
    Y_{t,i}=F_t^{-1}(u_{t,i}).
\]
Thus $u_{t,i}$ has conditional density $h_{\rm SM}(\cdot\mid q_t)$,
and $Y_{t,i}$ has conditional CDF $G_t=H_{\rm SM}\circ F_t$.
Equivalently, $Y_{t,i}=G_t^{-1}(V_{t,i})$.
For independent draws, the $V_{t,i}$ are conditionally independent uniforms.
For the common-factor design, draw
$a_i\sim U[-0.1,1.1]$ independently at the start of each replication and
hold the loadings fixed over time and across the five designs. Let
$C_t$ and $x_{t,i}$ be independent standard normal innovations, independent
over dates. Introduce $\gamma_t$ to adjust the strength of the common
component, as specified below, and define
\begin{align}
    \beta_{t,i}&=\frac{\gamma_ta_i}{\sqrt{1+(\gamma_ta_i)^2}},\\
    Z_{t,i}&=\beta_{t,i}C_t+
                 \sqrt{1-\beta_{t,i}^2}\,x_{t,i},\\
    u_{t,i}&=H_{\rm SM}^{-1}\{\Phi(Z_{t,i})\mid q_t\},
    \qquad Y_{t,i}=F_t^{-1}(u_{t,i}).
    \label{eq:simulation-copula}
\end{align}
Conditional on past information and the fixed loadings, every $Z_{t,i}$ is
standard normal. Thus the $V_{t,i}=\Phi(Z_{t,i})$ are identically uniform
but dependent across $i$ through the common factor $C_t$.
Applying the same two transformations gives each return the conditional
CDF $G_t$.
This construction relaxes independence without introducing heterogeneous
conditional marginals.

Before drawing the current innovations, $\gamma_t$ is calibrated to satisfy
\begin{equation}
    \frac1N\sum_{i=1}^N
    \frac{\operatorname{Var}\{E(Y_{t,i}\mid C_t,\mathcal F_{t-1})
                              \mid\mathcal F_{t-1}\}}
         {\operatorname{Var}(Y_{t,i}\mid\mathcal F_{t-1})}
    =0.10.
    \label{eq:simulation-share}
\end{equation}
The conditioning information includes the fixed loadings and the current
predictable density state.
The share concerns the transformed returns and averages individual explained
variance fractions. It is neither $\beta_{t,i}^2$, a common pairwise
correlation, nor a PCA eigenvalue share. Individual shares differ with
the loadings. In particular, the calculation integrates over individual
noise and common-factor realizations.

\subsection{Size and power}
\label{subsec:simulation-wald}

Estimation maximizes the time average of daily cross-sectional mean log
scores, as in \eqref{eq:dcd-objective}. Parameter inference uses the
lag-zero sandwich covariance in \eqref{eq:estimated-sandwich}. We test the three
coefficients of each block at the nominal five-percent level.
We jointly estimate the even and odd blocks and conduct a separate Wald
test of the three HAR coefficients in each block. Covariance calculations
hold the M2 estimates fixed for M3 and M4, and the M3 estimates fixed for
M5. The Fisher projection coefficients are also held fixed.

\input{tables/sec4_size}

Table~\ref{tab:sim-size} reports case 0, where every displayed test is a size
diagnostic. Under independence, rejection frequencies range from 2 to
8 percent. In the 10-percent factor design, they range from 2 to 12 percent.
M4 even rejects in 6, 12, and 5 percent of replications for $N=20,50,100$,
respectively. Dependence
does not produce uniformly severe overrejection in this design, but the
results do not justify nominal size control in every cell.
Each cell contains only 100 replications. At a true rejection probability
of 0.05, its Monte Carlo standard error is about 2.2 percentage points.

\input{tables/sec4_rejection}

Table~\ref{tab:sim-rejection} shows substantial detection of persistent
even dynamics. In case 2, M4 even rejects in 92, 100, and 100 percent of
the 10-percent-factor replications. Detection of tail order alone is
weaker: M3 rejects in 17, 30, and 46 percent in case 1. In case 4, M5 odd
rejects in 28, 34, and 68 percent. These are rejection frequencies at nominal
five-percent critical values, not size-adjusted power estimates.
The relative strengths of these components
therefore matter for finite-sample detection.

Rejection need not uniquely identify the generating component. In case 2,
M3 rejects despite static generating tail order, while in case 1 the
M4 even block can reject without a generating even update. A candidate
reference can leave another type of shape variation for its fitted block
to approximate. Such cells do not automatically diagnose the size of a
correctly specified component null. These results suggest that tail-order
and tail-intensity dynamics can be difficult to distinguish in finite samples.

\subsection{Forecast accuracy}
\label{subsec:simulation-oos}

The estimated coefficients are held fixed during the OOS period.
We evaluate the log score and one-percent quantile loss, averaged
across returns within each date, with M2 as the common benchmark.
M2 is estimated separately in each replication.

For each date $t$, the model forecasts one mixture density $\widehat g_t$
for the cross section. The daily log score is
\[
    \text{Log score}_t=\frac1N\sum_{i=1}^{N}\log\widehat g_t(Y_{t,i}).
\]
Higher values indicate better forecasts.

For a quantile level $\alpha\in(0,1)$ and the common predictive return
quantile $Q_t(\alpha)$, the quantile loss is
\[
    L_{t,i}(\alpha)
    =\bigl[\alpha-\mathbf{1}\{Y_{t,i}<Q_t(\alpha)\}\bigr]
      \bigl[Y_{t,i}-Q_t(\alpha)\bigr].
\]
Smaller values indicate better forecasts. This section uses $\alpha=0.01$,
and the empirical analysis in Section~\ref{subsec:oos-results} uses
$\alpha\in\{0.01,0.05\}$.

Within each replication, percentage improvements are
\[
\begin{aligned}
 \text{Loss improvement}
   &=100(\bar L_{\rm M2}-\bar L_m)/\bar L_{\rm M2},\\
 \text{Log-score improvement}
   &=100\{\exp(\overline{\log g_m-\log g_{\rm M2}})-1\}.
\end{aligned}
\]
We report the median percentage improvement over 100 replications.

\input{tables/sec4_oos_density_tail}

Table~\ref{tab:sim-oos-tail} reports log-score and one-percent quantile-loss
improvements for M3--M5 in cases 0--4. Median gains for M4
and M5 are positive in cases 2--4. The static design gives
essentially zero or slightly negative gains, as expected when the generating
distribution is static. Tail-order-only forecast gains are
small, and adding tail-intensity dynamics in that design does not
systematically improve on M2. Table~\ref{S-tab:sim-m5-m3-oos} of the
supplement additionally compares M5 with M3.

In case 4 with $N=100$ and the 10-percent factor, M5 reduces the
one-percent quantile loss by 0.360 percent at the median.
These results support adding tail-intensity dynamics, without implying a
uniform advantage of M5 over M4.
Detection of persistence and improvement in forecast scores are distinct
questions. The parameter tests assess shape dynamics, whereas the OOS
comparisons assess the predictive consequences of the fitted dynamics.

%% file: tables/sec4_size.tex
% Generated by build_sec4_sec5_tables.py from audited CSVs; do not edit numbers.
\begin{table}[htbp]
\centering
\caption{Empirical size of nominal five-percent Wald tests}
\label{tab:sim-size}
\fontsize{10pt}{12pt}\selectfont
\setlength{\tabcolsep}{4pt}
\begin{tabular}{lrrrrrr}
\hline
Dependence & $N$ & M3 $\eta$ & M4 even & M4 odd & M5 even & M5 odd\\
\hline
Independent & 20 & 3 & 6 & 8 & 6 & 8\\
Common 10\% & 20 & 4 & 6 & 7 & 5 & 7\\
Independent & 50 & 6 & 8 & 5 & 7 & 5\\
Common 10\% & 50 & 6 & 12 & 3 & 5 & 3\\
Independent & 100 & 3 & 5 & 8 & 2 & 8\\
Common 10\% & 100 & 3 & 5 & 4 & 2 & 5\\
\hline
\end{tabular}
\par\smallskip
\begin{minipage}{\textwidth}
\fontsize{9pt}{10.8pt}\selectfont All states are static (case 0). Entries are rejection percentages over 100 replications. All displayed tests are numerically valid. M1--M5 and the M1 Fisher geometry are estimated. The lag-zero covariance holds M2 estimates fixed for M3 and M4, and M3 estimates fixed for M5.
\end{minipage}
\end{table}

%% file: tables/sec4_rejection.tex
% Generated by build_sec4_sec5_tables.py from audited CSVs; do not edit numbers.
\begin{table}[htbp]
\centering
\caption{Rejection frequencies under persistent shape dynamics}
\label{tab:sim-rejection}
\fontsize{9pt}{10.8pt}\selectfont
\setlength{\tabcolsep}{4pt}
\begin{tabular}{llrrrrrr}
\hline
Case & Dependence & $N$ & M3 $\eta$ & M4 even & M4 odd & M5 even & M5 odd\\
\hline
1: $\eta$ only & Independent & 20 & 21 & 6 & 5 & 2 & 6\\
 & Common 10\% & 20 & 17 & 10 & 10 & 5 & 10\\
 & Independent & 50 & 48 & 13 & 4 & 9 & 4\\
 & Common 10\% & 50 & 30 & 12 & 7 & 8 & 7\\
 & Independent & 100 & 71 & 14 & 5 & 3 & 5\\
 & Common 10\% & 100 & 46 & 15 & 7 & 4 & 7\\
\hline
2: Even only & Independent & 20 & 36 & 97 & 10 & 93 & 10\\
 & Common 10\% & 20 & 26 & 92 & 8 & 90 & 8\\
 & Independent & 50 & 37 & 100 & 7 & 100 & 7\\
 & Common 10\% & 50 & 16 & 100 & 8 & 99 & 9\\
 & Independent & 100 & 65 & 100 & 6 & 100 & 6\\
 & Common 10\% & 100 & 13 & 100 & 6 & 100 & 6\\
\hline
3: $\eta$ and even & Independent & 20 & 22 & 97 & 8 & 97 & 8\\
 & Common 10\% & 20 & 26 & 97 & 8 & 92 & 7\\
 & Independent & 50 & 65 & 100 & 6 & 100 & 5\\
 & Common 10\% & 50 & 43 & 99 & 6 & 99 & 6\\
 & Independent & 100 & 93 & 100 & 5 & 100 & 5\\
 & Common 10\% & 100 & 58 & 100 & 3 & 100 & 3\\
\hline
4: $\eta$, even, and odd & Independent & 20 & 19 & 94 & 23 & 91 & 23\\
 & Common 10\% & 20 & 20 & 91 & 26 & 88 & 28\\
 & Independent & 50 & 59 & 100 & 57 & 100 & 57\\
 & Common 10\% & 50 & 46 & 98 & 34 & 97 & 34\\
 & Independent & 100 & 97 & 100 & 73 & 100 & 73\\
 & Common 10\% & 100 & 69 & 100 & 68 & 100 & 68\\
\hline
\end{tabular}
\par\smallskip
\begin{minipage}{\textwidth}
\fontsize{9pt}{10.8pt}\selectfont Percentages over 100 replications at nominal size 5\%. Columns for active components describe detection; rejections in other columns need not identify a separate true shape component. True coefficients depend on $N$ through the stated calibration. M1--M5 are estimated. Wald calculations hold M2 estimates fixed for M3 and M4, and M3 estimates fixed for M5.
\end{minipage}
\end{table}

%% file: tables/sec4_oos_density_tail.tex
% Generated by build_sec4_sec5_tables.py from audited CSVs; do not edit numbers.
\begin{table}[htbp]
\centering
\caption{Median OOS density and tail improvements relative to M2}
\label{tab:sim-oos-tail}
\fontsize{9pt}{10.8pt}\selectfont
\setlength{\tabcolsep}{4pt}
\begin{tabular}{llrrrrrrr}
\hline
\multicolumn{3}{l}{Criterion} & \multicolumn{3}{c}{Log score} & \multicolumn{3}{c}{1\% quantile loss}\\
\cline{4-6}\cline{7-9}
Case & Dependence & $N$ & M3 & M4 & M5 & M3 & M4 & M5\\
\hline
0: Static & Independent & 20 & -0.001 & -0.008 & -0.011 & -0.001 & -0.079 & -0.078\\
 & Common 10\% & 20 & -0.003 & -0.011 & -0.015 & +0.001 & -0.092 & -0.092\\
 & Independent & 50 & -0.002 & -0.004 & -0.006 & +0.001 & -0.033 & -0.035\\
 & Common 10\% & 50 & -0.002 & -0.006 & -0.007 & +0.002 & -0.053 & -0.049\\
 & Independent & 100 & -0.000 & -0.001 & -0.002 & -0.000 & -0.015 & -0.016\\
 & Common 10\% & 100 & -0.001 & -0.004 & -0.005 & -0.000 & -0.034 & -0.032\\
\hline
1: $\eta$ only & Independent & 20 & -0.001 & -0.007 & -0.007 & -0.001 & -0.014 & -0.012\\
 & Common 10\% & 20 & -0.002 & -0.009 & -0.011 & -0.004 & -0.064 & -0.071\\
 & Independent & 50 & +0.001 & -0.004 & -0.003 & +0.000 & -0.038 & -0.033\\
 & Common 10\% & 50 & +0.001 & -0.005 & -0.005 & -0.003 & -0.045 & -0.042\\
 & Independent & 100 & +0.002 & -0.001 & +0.001 & +0.001 & -0.011 & -0.011\\
 & Common 10\% & 100 & +0.002 & -0.003 & -0.003 & -0.001 & -0.041 & -0.038\\
\hline
2: Even only & Independent & 20 & +0.000 & +0.033 & +0.030 & +0.032 & +0.296 & +0.286\\
 & Common 10\% & 20 & +0.000 & +0.035 & +0.028 & +0.034 & +0.356 & +0.396\\
 & Independent & 50 & +0.002 & +0.034 & +0.030 & +0.023 & +0.312 & +0.309\\
 & Common 10\% & 50 & +0.000 & +0.033 & +0.028 & +0.024 & +0.365 & +0.367\\
 & Independent & 100 & +0.001 & +0.027 & +0.025 & +0.020 & +0.301 & +0.298\\
 & Common 10\% & 100 & +0.000 & +0.031 & +0.030 & +0.017 & +0.319 & +0.313\\
\hline
3: $\eta$ and even & Independent & 20 & +0.001 & +0.033 & +0.031 & +0.019 & +0.312 & +0.333\\
 & Common 10\% & 20 & +0.000 & +0.027 & +0.026 & +0.019 & +0.111 & +0.120\\
 & Independent & 50 & +0.002 & +0.026 & +0.026 & +0.014 & +0.255 & +0.254\\
 & Common 10\% & 50 & +0.001 & +0.031 & +0.029 & +0.014 & +0.298 & +0.310\\
 & Independent & 100 & +0.004 & +0.028 & +0.030 & +0.013 & +0.255 & +0.256\\
 & Common 10\% & 100 & +0.003 & +0.038 & +0.040 & +0.012 & +0.358 & +0.368\\
\hline
4: $\eta$, even, and odd & Independent & 20 & -0.001 & +0.037 & +0.036 & +0.010 & +0.293 & +0.272\\
 & Common 10\% & 20 & +0.001 & +0.031 & +0.027 & +0.009 & +0.378 & +0.382\\
 & Independent & 50 & +0.003 & +0.033 & +0.033 & +0.014 & +0.329 & +0.335\\
 & Common 10\% & 50 & +0.001 & +0.043 & +0.045 & +0.012 & +0.337 & +0.347\\
 & Independent & 100 & +0.003 & +0.036 & +0.035 & +0.013 & +0.298 & +0.294\\
 & Common 10\% & 100 & +0.003 & +0.048 & +0.055 & +0.005 & +0.353 & +0.360\\
\hline
\end{tabular}
\par\smallskip
\begin{minipage}{\textwidth}
\fontsize{9pt}{10.8pt}\selectfont Median percentage improvements over 100 replications. Within each replication, log-score improvements are $100\{\exp(\overline{\Delta\log g})-1\}$ and quantile-loss improvements are percentage reductions relative to M2. All models, including the M2 benchmark, are estimated.
\end{minipage}
\end{table}

%% file: sec5.tex
\section{Empirical evidence}
\label{sec:empirical}

This section studies conditional density dynamics in 19 equity cross
sections: the S\&P 500 and Nikkei 225 constituent panels, 11 US sectors,
and six Japanese sectors. We use the symmetric-$t$ M1--M5 specification
and compare raw returns with residuals obtained by removing a first
principal component estimated once from the training sample and then
held fixed separately within each panel. The
broad-market panels contain constituent returns, not single index returns.
``Raw'' and ``PC1 residual'' use the same observed constituent--dates,
restricted to constituents with estimable training PC1 loadings.
We use $w_{t,i}=1/N_t$ throughout, so the target is the equally weighted
conditional mixture $\bar g_t=N_t^{-1}\sum_i g_{t,i}$, a special case of
the weighted mixture defined in Section~\ref{subsec:description}.
Constituent distributions may differ,
and returns within a date may be dependent.

\subsection{Data and sample construction}
\label{subsec:data-sample}

The common core sample begins on July 5, 2016, except for the S\&P 500
Real Estate sector, which begins on September 1, 2016. Estimation ends
on December 29, 2023. Strict one-step-ahead OOS evaluation runs through
July 24, 2026, beginning on January 2, 2024 for the US panels and
January 4, 2024 for the Japanese panels. There are 642 US and 624 Japanese
OOS trading dates. Table~\ref{S-tab:empirical-samples} in the supplement
reports panel sizes and the number of training dates.

Daily prices are obtained from
\href{https://finance.yahoo.com/}{Yahoo Finance}. We use
adjusted closing prices when available and otherwise reported closing
prices, and calculate returns as daily log differences. Raw returns are
not normalized by realized means or variances.

Membership is reconstructed point in time from the community-maintained
\href{https://github.com/fja05680/sp500}{\texttt{fja05680/sp500}}
S\&P 500 histories and the official Nikkei
\href{https://indexes.nikkei.co.jp/nkave/archives/file/nikkei_stock_average_weight_en.csv}{constituent list}
and
\href{https://indexes.nikkei.co.jp/nkave/archives/file/history_of_nikkei_stock_average_component_changes_en.pdf}{component-change history}.
Sector panels preserve
these histories. Source and classification details are given in
Section~\ref{S-supp:data-construction} of the supplement. Price data are
unavailable for some historical constituents, including delisted stocks.
On each date, we use only constituents with available returns.
Section~\ref{sec:theory} assumes that
inclusion and $N_t$ are determined before the current returns are observed.
Membership reconstruction alone does not establish this condition for
missing observations.

Median cross-sectional sizes range from 18 to 70 among sectors, compared
with 455 and 218 in the broad-market panels. Some sector samples contain
fewer than 20 observations on individual dates.

\subsection{Symmetric reference, estimation, and inference}
\label{subsec:empirical-specification}

M1 estimates a static symmetric standardized Student-$t$ reference, with
location, scale, and tail-order parameter $\eta$ free.
M2 fixes the fitted M1 tail order and estimates six HAR$(1,5,22)$
coefficients for location and log variance. M3 holds the M2 estimates of
$m_t$ and $\sigma_t$ fixed for each date and estimates three HAR
coefficients for latent tail order $\xi_\eta$. No parametric skewness
dynamics are estimated.

M4 uses the same M2 estimates and estimates six even/odd DCD-SM
coefficients. M5 holds the M3 estimates of $m_t$, $\sigma_t$, and $\eta_t$
fixed for each date and estimates six DCD-SM coefficients. The Fisher
projection coefficients are computed once at $q=0$ and the fitted M1
reference values, and held fixed. M4 computes its PIT observations and
reference scores using the M2 estimates for each date, and M5 uses the M3
estimates. The spline scores are orthogonalized with respect to the
location and log-variance scores, while they are not orthogonalized with respect
to the $\eta$ score.
Although the references are symmetric, the odd spline correction allows
the final densities to be asymmetric.

Every estimation stage uses the same equal-date criterion,
\begin{equation}
    L_T(\theta)=\frac1T\sum_{t=1}^T\frac1{N_t}
                       \sum_{i=1}^{N_t}\log\bar g_t(r_{t,i}\mid\theta).
    \label{eq:empirical-daily-mean}
\end{equation}
Each date receives equal weight as the number of covered constituents
changes over time. HAR windows contain exactly 1, 5, and 22 previous dates, with
unobserved presample scores initialized to zero.

In-sample inference uses fitted-alternative conditional sandwich-Wald
tests with lag zero and the covariance in \eqref{eq:estimated-sandwich}.
The fitted reference distribution and Fisher projection coefficients are
held fixed in each covariance calculation. For residual panels, the fitted
PCA means and loadings are also held fixed. Their estimation uncertainty
is not added to the parameter covariance.
In the tables, ``$\eta$'' denotes the latent tail-order block
$\xi_\eta$, whose HAR coefficients are defined in \eqref{eq:t-latent-har}.
Each of location, scale, $\eta$, even, and odd has a three-coefficient
$\chi^2_3$ Wald test. The M2 location test sets the three location HAR
coefficients to zero, retaining fitted log-variance dynamics. The scale
test sets the three log-variance HAR coefficients to zero, retaining
fitted location dynamics.
The M4/M5 joint null has six coefficients. Detection
counts use the nominal five-percent level without adjustment for multiple
tests.

\subsection{Persistence in the raw cross sections}
\label{subsec:market-wide-results}

Table~\ref{tab:market-wald} gives broad-market $p$-values for raw returns
and PC1 residuals. In the S\&P 500 raw panel, the M3 $\eta$ test has
$p=5.06\times10^{-5}$, while M4 and M5 even have
$p=1.38\times10^{-9}$ and $3.46\times10^{-19}$.
For the Nikkei 225, the corresponding $p$-values are
$3.83\times10^{-4}$, $2.81\times10^{-4}$, and
$4.20\times10^{-4}$. Odd persistence is detected in the
S\&P 500 panel, with $p=0.0201$ for M4 and $p=0.0456$ for M5,
but not in the Nikkei 225 raw panel.

\input{tables/sec5_market_wald}
\input{tables/sec5_wald_counts}

Across the 19 raw panels, M2 detects location dynamics in four and scale
dynamics in all 19. Among the 17 sectors, the corresponding counts are
three and 17. M3 detects $\eta$ dynamics in 15 panels, including 13 sectors,
whereas M4 and M5 detect even dynamics in all 19 panels.
Odd detections are less frequent:
five panels for M4 and three for M5, including four and two sectors.
The evidence is consequently strongest for changes in both tail intensities
together. The broad-market and sector samples overlap. These counts are
descriptive summaries, not independent replications of the same test.

The $\eta$ and even results are related descriptions of changing tails,
not separate economic mechanisms established by the tests.
The M5 even results support adding tail-intensity dynamics to the fitted
M3 reference.

\subsection{Removing a within-panel principal component}
\label{subsec:pca-robustness}

Our main adjustment removes one common linear component separately in each
of the 19 panels. We estimate a covariance PCA once using all training
dates through December 29, 2023, and hold its means and leading loading
vector fixed throughout training and OOS evaluation. Each constituent is
demeaned by its observed training mean, without volatility standardization.
For PCA only, missing training returns are replaced by that mean, or
equivalently by zero after demeaning. Constituents without nonconstant
observed training returns are excluded. No imputed return enters the
density estimation or forecast evaluation. This training-sample
preprocessing is retrospective within the training period. It uses no OOS
information and is not reestimated during OOS evaluation.

Let $\widehat\mu$ and $\widehat b$ denote the fixed means and loadings,
and let $\mathcal O_t$ identify eligible constituents observed at date $t$.
For their return vector $y_t$, define
\[
    \widehat F_t=
       \frac{\widehat b_{\mathcal O_t}'(y_t-\widehat\mu_{\mathcal O_t})}
            {\widehat b_{\mathcal O_t}'\widehat b_{\mathcal O_t}},\qquad
    e_t=y_t-\widehat\mu_{\mathcal O_t}
            -\widehat b_{\mathcal O_t}\widehat F_t.
\]
The subscript $\mathcal O_t$ selects entries of the fixed vectors.
The residual removes both the training
mean and the projected PC1 component. We fit all five models to raw returns
and residuals on the same constituent--dates. The two targets give 38
panels and 190 fitted models. There is no initial rolling-window exclusion.

The broad-market mean pairwise correlations decline from 0.356 to 0.009
for the S\&P 500 and from 0.361 to 0.002 for the Nikkei 225. The leading
eigenvalues account for 36.5 and 37.5 percent of their respective training
covariance traces. These PCA shares are
not the conditional explained-return variance fractions calibrated in
Section~\ref{subsec:simulation-dependence}. Construction details are given
in Section~\ref{S-supp:pca-details} of the supplement.

M2 detects location dynamics in three residual panels, including two
sectors, and scale dynamics in all 19 residual panels.
Tail-intensity dynamics remain after this removal. Both M4 and M5 detect
even persistence in all 19 residual panels, as they do in the raw
panels. M3 $\eta$ is detected in 17 of 19 residual panels.
Odd persistence is detected in five residual panels for M4 and in five
for M5. Both models detect odd persistence in the S\&P 500 and Nikkei 225
residual panels. For the US
Communication Services residuals, the fitted M1 latent tail-order
parameter reaches its lower search bound, giving $\eta=2.0506$.
The M3, M4, and M5 persistence tests for this panel are numerically valid.

\subsection{Strict out-of-sample forecasting}
\label{subsec:oos-results}

Having documented shape dynamics through the proposed conditional tests,
we next examine their consequences for OOS forecast accuracy. These
comparisons address the predictive gains from fitting the dynamics,
separately from testing whether the shape-dynamics coefficients are zero.
We compare M4 with M2 and M5 with M3 separately for raw returns and PC1
residuals. These comparisons assess the incremental contribution of
tail-intensity dynamics relative to their respective references. The M5--M3 comparison
matches the simulation comparison in
Table~\ref{S-tab:sim-m5-m3-oos} of the supplement.

All fitted coefficients, Fisher geometry, and PCA preprocessing are held fixed
at the training estimates. After observing returns on date $t-1$, we
update the dynamic parameters to construct the density forecast for date
$t$. We evaluate the log score and one- and five-percent quantile losses,
using the definitions in Section~\ref{subsec:simulation-oos}.
Higher log scores and lower losses indicate better forecasts.
We average the daily log scores and quantile losses over the OOS period.

\input{tables/sec5_oos_counts}
\input{tables/sec5_market_oos}

Table~\ref{tab:oos-pairwise} counts positive panel-level improvements and,
in parentheses, positive differences with two-sided forecast-comparison
$p<0.05$. These forecast-difference standard errors retain HAC(22)
\citep{NeweyWest1987}. They are separate from the lag-zero parameter Wald
covariances.

For raw returns, M4 improves continuous log score relative to M2 in
11 of 19 panels, while M5 improves on M3 in 14. Their one-percent
quantile-loss improvement counts are 14 and 15, and their five-percent
counts are eight and thirteen.

For the training-fixed PC1 residuals, M4 and M5 improve log score relative
to their respective references in 13 and 14 panels, with four significant
positive differences each. One-percent quantile loss improves in 15 and
16 residual panels, with three significant positive differences each.
Five-percent quantile loss improves in ten and eleven panels, without
significant positive differences.

Broad-market magnitudes in Table~\ref{tab:oos-market} illustrate these
differences. In raw Nikkei 225 returns, M4 and M5 improve log score by
0.263 and 0.300 percent relative to M2 and M3 and reduce one-percent quantile
loss by approximately 6.98 and 7.03 percent.
For S\&P 500 PC1 residuals, strong parameter-test detections coexist with
small negative log-score differences for both
corrections. The M5 one-percent quantile-loss reduction is only about
0.009 percent. For Nikkei 225 residuals, the corresponding M4 and M5 loss
reductions are 1.53 and 1.81 percent.
In-sample persistence detection does not ensure a material improvement in
each OOS score.

Removing a common component also removes any predictable variation it
contains and changes the forecast target. We therefore do not require
the incremental gains to increase after residualization. The relevant
finding is that additional tail-intensity dynamics and some forecast gains
remain after the main contemporaneous linear component has been removed.

\subsection{Empirical interpretation}
\label{subsec:empirical-interpretation}

The empirical evidence supports persistent tail-intensity dynamics in
equity cross-sectional mixture densities. The even component describes
changes in left and right tail intensities in the same direction relative
to the fitted reference. Asymmetry in a return distribution does not by
itself imply persistent changes in its asymmetry. In both M4 and M5,
persistence in the imbalance between left and right tail intensities is
detected in only a minority of panels, whereas same-direction tail-intensity
dynamics are detected in all 19 raw panels and all 19 PC1-residual panels.
In M5, the even component remains significant after fitting
location, scale, and tail-order dynamics. These results support adding
tail-intensity dynamics to the reference model, including after the main
within-panel linear common component has been removed.

The OOS results show that fitting these dynamics can also improve density
and lower-tail quantile forecasts. Relative to M2 and M3, respectively,
M4 and M5 improve log score and one-percent quantile loss in a majority
of both raw and PC1-residual panels. The size of these gains varies across
panels, and statistically significant gains occur less frequently than
positive point estimates. Five-percent quantile-loss improvements are less
consistent. The predictive contribution of tail-intensity dynamics is
therefore most evident in the log score and the lower one-percent quantile
loss, with gains remaining in many panels after PC1 removal.

Taken together, the persistence tests and forecast comparisons suggest that
the detected tail-intensity dynamics have predictive value. Persistence is
detected in all panels, and including these dynamics improves log score
and one-percent quantile loss in most panels, both before and after PC1
removal. The detection of persistent tail variation is thus accompanied
by a tendency toward better OOS forecasts. The size of the forecast gains
depends on the panel and the criterion, as illustrated by the small gains
for the S\&P 500 residuals despite strong persistence detections.

%% file: tables/sec5_market_wald.tex
% Generated by build_sec4_sec5_tables.py from audited CSVs; do not edit numbers.
\begin{table}[htbp]
\centering
\caption{Broad-market conditional sandwich-Wald $p$-values}
\label{tab:market-wald}
\fontsize{10pt}{12pt}\selectfont
\setlength{\tabcolsep}{4pt}
\begin{tabular}{lrrrrrr}
\hline
Model & Location & Scale & $\eta$ & Even & Odd & Joint\\
\hline
\multicolumn{7}{l}{Raw; S\&P 500}\\
M2 & 0.0095 & $6.91\times10^{-49}$ & -- & -- & -- & --\\
M3 & -- & -- & $5.06\times10^{-5}$ & -- & -- & --\\
M4 & -- & -- & -- & $1.38\times10^{-9}$ & 0.0201 & $1.34\times10^{-14}$\\
M5 & -- & -- & -- & $3.46\times10^{-19}$ & 0.0456 & $2.99\times10^{-21}$\\
\hline
\multicolumn{7}{l}{Raw; Nikkei 225}\\
M2 & 0.6742 & $1.33\times10^{-20}$ & -- & -- & -- & --\\
M3 & -- & -- & $3.83\times10^{-4}$ & -- & -- & --\\
M4 & -- & -- & -- & $2.81\times10^{-4}$ & 0.4576 & $3.74\times10^{-5}$\\
M5 & -- & -- & -- & $4.20\times10^{-4}$ & 0.3712 & $1.07\times10^{-4}$\\
\hline
\multicolumn{7}{l}{PC1 residual; S\&P 500}\\
M2 & 0.3554 & $3.73\times10^{-81}$ & -- & -- & -- & --\\
M3 & -- & -- & $8.56\times10^{-12}$ & -- & -- & --\\
M4 & -- & -- & -- & $4.55\times10^{-7}$ & $2.37\times10^{-4}$ & $1.69\times10^{-8}$\\
M5 & -- & -- & -- & $6.33\times10^{-6}$ & $4.64\times10^{-4}$ & $4.54\times10^{-7}$\\
\hline
\multicolumn{7}{l}{PC1 residual; Nikkei 225}\\
M2 & 0.0011 & $2.54\times10^{-100}$ & -- & -- & -- & --\\
M3 & -- & -- & $7.20\times10^{-14}$ & -- & -- & --\\
M4 & -- & -- & -- & $3.51\times10^{-106}$ & $2.12\times10^{-5}$ & $5.50\times10^{-117}$\\
M5 & -- & -- & -- & $8.59\times10^{-104}$ & $5.04\times10^{-5}$ & $3.76\times10^{-113}$\\
\hline
\end{tabular}
\par\smallskip
\begin{minipage}{\textwidth}
\fontsize{9pt}{10.8pt}\selectfont The reference is symmetric Student-$t$. Each M2 test sets the three HAR coefficients for the indicated component to zero, retaining fitted dynamics of the other component. M3 has one three-coefficient $\eta$ block. M4 and M5 have even and odd three-coefficient blocks and a six-coefficient joint null. All parameter covariances use lag zero and hold M1 estimates fixed for M2, M2 estimates fixed for M3 and M4, and M3 estimates fixed for M5.
\end{minipage}
\end{table}

%% file: tables/sec5_wald_counts.tex
% Generated by build_sec4_sec5_tables.py from audited CSVs; do not edit numbers.
\begin{table}[htbp]
\centering
\caption{Persistence rejections at the five-percent level}
\label{tab:wald-counts}
\fontsize{10pt}{12pt}\selectfont
\setlength{\tabcolsep}{4pt}
\begin{tabular}{llrrrrrrr}
\hline
\rule{0pt}{30pt}Target & Universe & \shortstack{M2\\location} & \shortstack{M2\\scale} & \shortstack{M3\\$\eta$} & \shortstack{M4\\even} & \shortstack{M4\\odd} & \shortstack{M5\\even} & \shortstack{M5\\odd}\\
\hline
Raw & All panels & 4/19 & 19/19 & 15/19 & 19/19 & 5/19 & 19/19 & 3/19\\
Raw & Sectors & 3/17 & 17/17 & 13/17 & 17/17 & 4/17 & 17/17 & 2/17\\
PC1 residual & All panels & 3/19 & 19/19 & 17/19 & 19/19 & 5/19 & 19/19 & 5/19\\
PC1 residual & Sectors & 2/17 & 17/17 & 15/17 & 17/17 & 3/17 & 17/17 & 3/17\\
\hline
\end{tabular}
\par\smallskip
\begin{minipage}{\textwidth}
\fontsize{9pt}{10.8pt}\selectfont Raw and PC1 residual use the same constituent--dates. Counts are descriptive and unadjusted for multiplicity. All displayed tests are numerically valid. Numerical validity of a Wald statistic does not establish exact finite-sample size.
\end{minipage}
\end{table}

%% file: tables/sec5_oos_counts.tex
% Generated by build_sec4_sec5_tables.py from audited CSVs; do not edit numbers.
\begin{table}[htbp]
\centering
\caption{Panels with positive OOS forecast improvements}
\label{tab:oos-pairwise}
\fontsize{10pt}{12pt}\selectfont
\setlength{\tabcolsep}{4pt}
\begin{tabular}{llrrr}
\hline
Target & Comparison & Log score & 1\% loss & 5\% loss\\
\hline
Raw & M4--M2 & 11 (2) & 14 (0) & 8 (0)\\
Raw & M5--M3 & 14 (4) & 15 (2) & 13 (0)\\
Raw & M5--M4 & 18 (1) & 14 (0) & 17 (1)\\
PC1 residual & M4--M2 & 13 (4) & 15 (3) & 10 (0)\\
PC1 residual & M5--M3 & 14 (4) & 16 (3) & 11 (0)\\
PC1 residual & M5--M4 & 14 (5) & 13 (1) & 17 (4)\\
\hline
\end{tabular}
\par\smallskip
\begin{minipage}{\textwidth}
\fontsize{9pt}{10.8pt}\selectfont Each entry is the number of positive point-estimate improvements out of 19 panels, followed by the number with a two-sided HAC(22) forecast-difference $p$-value below 0.05 in parentheses. The HAC correction here concerns forecast differences, not the lag-zero parameter Wald tests.
\end{minipage}
\end{table}

%% file: tables/sec5_market_oos.tex
% Generated by build_sec4_sec5_tables.py from audited CSVs; do not edit numbers.
\begin{table}[htbp]
\centering
\caption{Broad-market OOS improvements within each target}
\label{tab:oos-market}
\fontsize{10pt}{12pt}\selectfont
\setlength{\tabcolsep}{4pt}
\begin{tabular}{lllrr}
\hline
Target & Market & Comparison & Log score (\%) & 1\% loss (\%)\\
\hline
Raw & S\&P 500 & M4--M2 & +0.292 & -0.169\\
Raw & S\&P 500 & M5--M3 & +0.220 & +0.837\\
Raw & Nikkei 225 & M4--M2 & +0.263 & +6.981\\
Raw & Nikkei 225 & M5--M3 & +0.300 & +7.031\\
PC1 residual & S\&P 500 & M4--M2 & -0.005 & -0.065\\
PC1 residual & S\&P 500 & M5--M3 & -0.018 & +0.009\\
PC1 residual & Nikkei 225 & M4--M2 & +0.047 & +1.527\\
PC1 residual & Nikkei 225 & M5--M3 & +0.101 & +1.810\\
\hline
\end{tabular}
\par\smallskip
\begin{minipage}{\textwidth}
\fontsize{9pt}{10.8pt}\selectfont Log-score improvements are $100\{\exp(\overline{\Delta\log g})-1\}$, using the proposed model minus the benchmark. Loss improvements are percentage reductions relative to the indicated benchmark. Positive values improve forecasting.
\end{minipage}
\end{table}

%% file: sec6.tex
\section{Conclusion}
\label{sec:conclusion}

This paper has developed dynamic conditional densities with spline mixtures
(DCD-SM) for a conditional mixture density chosen by the analyst. The model
is fitted directly to cross-sectional returns using a weighted log score.
Individual conditional densities may differ, and dependence within each
date need not be specified. The spline correction allows left and right
tail intensities to change while preserving the reference tail order.
With the reference distribution and Fisher projection matrices known and
the spline dimension fixed, we establish consistency and asymptotic
normality of the dynamic coefficient estimators.

Simulations evaluate the finite-sample size and power of the Wald tests
and the accuracy of out-of-sample forecasts under independent and dependent
cross-sectional sampling. In US and Japanese equity cross sections, we
detect persistent changes that affect left and right tail intensities in
the same direction, even after fitting dynamic location, scale, and tail
order and removing the leading principal component. Adding tail-intensity
dynamics improves log scores and one-percent lower-tail quantile losses
in most panels, although the magnitude and statistical significance of
these gains vary.

The same spline-mixture representation can also be used in density-valued
VAR models fitted to previously estimated densities. Those models treat
the estimated densities as observations and place stochastic innovations
in their dynamics. DCD-SM instead fits the chosen conditional mixture
density directly to returns and updates its parameters from past
observations through a score-driven recursion, without separate
innovations in the density dynamics. Introducing such innovations is a
natural extension. This would make the density parameters latent and
require a model linking them to cross-sectional returns, together with
a filtering method, potentially using particle filtering.

%% file: backmatter.tex
\section*{Funding}
This work was supported by the Japan Society for the Promotion of Science
KAKENHI [grant number 26H01952].

\section*{Acknowledgements}
\input{ai_disclosure}

\section*{Data availability}
The data sources and sample construction are described in
Section~\ref{subsec:data-sample} and
Section~\ref{S-supp:data-construction} of the supplement.
The simulation code, empirical estimation
code, table-generation code, and aggregated numerical results are available
in the companion repository at
\url{https://github.com/matsuda318/DCD-SM}. The repository records
sample construction and input hashes; the simulations use synthetic returns.
Empirical input files and intermediate daily paths are retained locally
and are not included in the public repository.

\section*{Conflict of interest}
The authors declare no conflicts of interest.

%% file: ai_disclosure.tex
The first author used OpenAI GPT-6.1 Sol to assist with manuscript revision and
code development. The first author directed these tasks and retained responsibility
for checking the generated text and code and for the final manuscript and results.

%% file: arxiv_supplement.tex
% Generated by build_arxiv.ps1; edit the shared supp_*.tex sources.

% Source: supp_fisher.tex
\section{Fisher projections at the symmetric \texorpdfstring{Student-$t$}{Student-t} anchor}
\label{S-supp:fisher}

Let $(m_0,\sigma_0,\eta_0)$ be the fixed M1 reference parameters, with
$\sigma_0>0$ and $\eta_0>2$, and let $v_0=\log\sigma_0^2$. At $q_0=0$, the
correction density is one, so the corrected return density equals the
variance-standardized Student-$t$ reference in
\eqref{eq:t-reference-density}. Write $E_0$ for expectation under this
reference. The Fisher blocks below are expectations of single-observation
score products. They define the fixed projection applied to the weighted
date-average scores in the main paper. The formulas are evaluated
separately at each raw-return or residual M1 estimate.

\subsection{Reference Fisher blocks and the M3 projection}

For $z=(r-m_0)/\sigma_0$, the location and log-variance scores at the anchor are
\begin{equation}
\begin{aligned}
    \rho_m(z)
    &=\frac{(\eta_0+1)z}{\sigma_0(\eta_0-2+z^2)},\\
    \rho_v(z)
    &=-\frac12+\frac{(\eta_0+1)z^2}{2(\eta_0-2+z^2)}.
\end{aligned}
\label{S-eq:anchor-ms-scores}
\end{equation}
The location score is odd in $z$, whereas the log-variance and
tail-order scores are even. Hence $I_{m,v,0}=I_{\eta,m,0}=0$.
The remaining entries give
\begin{equation}
    I_{ms,ms,0}
    =\begin{pmatrix}
       \dfrac{\eta_0(\eta_0+1)}
       {\sigma_0^2(\eta_0-2)(\eta_0+3)} & 0\\[6pt]
       0 & \dfrac{\eta_0}{2(\eta_0+3)}
     \end{pmatrix}.
\label{S-eq:closed-ms-information}
\end{equation}
With $\dot\eta_0$ as defined in the main paper, the latent tail-order
cross-information is
\begin{equation}
    I_{\xi_\eta,ms,0}
    =\left(
       0,\;
       \frac{3\dot\eta_0}
       {(\eta_0-2)(\eta_0+1)(\eta_0+3)}
      \right).
\label{S-eq:closed-eta-ms-information}
\end{equation}
Multiplying \eqref{S-eq:closed-eta-ms-information} by the inverse of
\eqref{S-eq:closed-ms-information} gives
\eqref{eq:t-explicit-projection}. These blocks and the resulting M3
projection coefficient can therefore be computed without numerical
integration.

To derive these expressions, let $Z$ have the standardized anchor
distribution and set
\[
    X=\frac{Z^2}{\eta_0-2+Z^2}.
\]
Then $X\sim\operatorname{Beta}(1/2,\eta_0/2)$, with
\[
    E_0X=\frac1{\eta_0+1},\qquad
    \operatorname{Var}_0(X)
    =\frac{2\eta_0}{(\eta_0+1)^2(\eta_0+3)}.
\]
Also,
\[
    E_0\{X(1-X)\}
    =\frac{\eta_0}{(\eta_0+1)(\eta_0+3)}.
\]
Substitution into $E_0\rho_m(Z)^2$ and $E_0\rho_v(Z)^2$ gives
\eqref{S-eq:closed-ms-information}. For the cross-information, differentiate
the log-variance score with respect to $\eta$ while holding $r,m_0,\sigma_0$
fixed. The information identity yields
\begin{align*}
    I_{\eta,v,0}
    &=-E_0\left[
       \left.\frac{\partial\rho_v(Z;\eta)}{\partial\eta}
       \right|_{\eta=\eta_0}\right]\\
    &=-\frac1{2(\eta_0+1)}
       +\frac{\eta_0}{2(\eta_0-2)(\eta_0+3)}\\
    &=\frac{3}{(\eta_0-2)(\eta_0+1)(\eta_0+3)}.
\end{align*}
The chain rule $\rho_{\xi_\eta}=\dot\eta_0\rho_\eta$ then gives
\eqref{S-eq:closed-eta-ms-information}.

\subsection{Spline cross-information and the M4--M5 projections}

Let $B(u)=(B_1(u),\ldots,B_J(u))'$ and let $m_j=\int_0^1 B_j(u)\,du$
denote the spline masses, distinct from the scalar location $m_0$.
At $q_0=0$, the mixture weights equal the mass vector, and the
responsibilities equal $B(u)$. The single-observation spline score is thus
\[
    \psi_q(u,0)=Q\{B(u)-(m_1,\ldots,m_J)'\}.
\]
Let $z_0(u)$ be the quantile function of the mean-zero, unit-variance
Student-$t$ distribution with $\eta_0$ degrees of freedom. Since the
reference PIT is uniform, the spline cross-information is
\begin{equation}
    I_{q,ms,0}
    =\int_0^1\psi_q(u,0)
       \begin{pmatrix}\rho_m\{z_0(u)\}&\rho_v\{z_0(u)\}\end{pmatrix}
       \,du.
\label{S-eq:spline-ms-cross-information}
\end{equation}
This is an explicit one-dimensional integral. For general spline knots,
its entries are evaluated numerically because they involve the
Student-$t$ quantile function together with the spline basis.

For a symmetric knot sequence, $B_j(1-u)=B_{J+1-j}(u)$ and
$m_j=m_{J+1-j}$. The even and odd scores at the anchor simplify to
\begin{equation}
\begin{aligned}
    e_{\rm even}'\psi_q(u,0)
    &=B_1(u)+B_J(u)-m_1-m_J,\\
    e_{\rm odd}'\psi_q(u,0)
    &=B_J(u)-B_1(u).
\end{aligned}
\label{S-eq:anchor-endpoint-scores}
\end{equation}
The first is even about $u=1/2$ and the second is odd. Also,
$z_0(1-u)=-z_0(u)$, so
\[
    e_{\rm even}'I_{q,m,0}=0,\qquad
    e_{\rm odd}'I_{q,v,0}=0.
\]
The two remaining cross-information entries are
\begin{align}
    e_{\rm even}'I_{q,v,0}
    &=\int_0^1\{B_1(u)+B_J(u)-m_1-m_J\}
       \rho_v\{z_0(u)\}\,du,\label{S-eq:even-v-cross-information}\\
    e_{\rm odd}'I_{q,m,0}
    &=\int_0^1\{B_J(u)-B_1(u)\}
       \rho_m\{z_0(u)\}\,du.\label{S-eq:odd-m-cross-information}
\end{align}
Combining these entries with \eqref{S-eq:closed-ms-information} gives the
coefficients in \eqref{eq:endpoint-projection-coefficients} and the
scalar updates in \eqref{eq:explicit-endpoint-projection}.
M4 and M5 use the same projection coefficients, computed once at $q=0$
and the fixed M1 reference values. Their PIT observations and reference
scores are computed at each date using the M2 and M3 reference
distributions, respectively.
The symmetry cancellations concern the fixed anchor and do not require
the corrected density at every date to be symmetric.

In the saved experiments, both the reference Fisher blocks and the spline
cross-information are computed by Gauss--Legendre quadrature with 128 nodes
per PIT subinterval. The reference calculation subdivides the probability
interval near the tails, and the spline calculation uses the knot spans.
The closed forms above identify the same reference quantities and provide
analytic checks.

% Source: supp_endpoints.tex
\section{Proof of the tail-intensity limit}
\label{S-supp:endpoints}

This section proves Proposition~\ref{thm:endpoint-intensity}. We first
derive the formulas for fixed $J$, then let the supports of the first and
last B-splines shrink as $J$ increases. The amplitude sequences are held fixed in this
calculation. The estimation theory in Section~\ref{sec:theory} uses a
fixed $J$.

\subsection{Formulas for fixed J}

Let $B_{1,J},\ldots,B_{J,J}$ be a clamped B-spline basis of fixed degree on
$[0,1]$, and let
\[
 m_{j,J}=\int_0^1 B_{j,J}(u)\,du.
\]
Write $e_j$ for the $j$th coordinate vector in $\mathbb R^J$. With the
Helmert matrix $Q_J$, define
\[
 e_{{\rm even},J}=Q_J(e_1+e_J),\qquad
 e_{{\rm odd},J}=Q_J(e_J-e_1).
\]
For given amplitudes, write
\begin{equation}
 q_{t,J}=a_{{\rm even},t}e_{{\rm even},J}
          +a_{{\rm odd},t}e_{{\rm odd},J},
 \label{S-eq:endpoint-state-J}
\end{equation}
and set
\[
 \alpha_{L,t}=a_{{\rm even},t}-a_{{\rm odd},t},\qquad
 \alpha_{R,t}=a_{{\rm even},t}+a_{{\rm odd},t}.
\]

\begin{lemma}
\label{S-thm:endpoint-intensity}
For each $J$, the mixture weights are
\begin{equation}
\begin{aligned}
 p_{1,t,J}&=\frac{m_{1,J}e^{\alpha_{L,t}}}{Z_{t,J}},&
 p_{J,t,J}&=\frac{m_{J,J}e^{\alpha_{R,t}}}{Z_{t,J}},\\
 p_{j,t,J}&=\frac{m_{j,J}}{Z_{t,J}},&
 &j=2,\ldots,J-1,
\end{aligned}
 \label{S-eq:endpoint-weights-J}
\end{equation}
where
\begin{equation}
 Z_{t,J}=1+m_{1,J}(e^{\alpha_{L,t}}-1)
           +m_{J,J}(e^{\alpha_{R,t}}-1).
 \label{S-eq:endpoint-Z-J}
\end{equation}
The correction density is
\begin{equation}
 h_J(u\mid q_{t,J})
  =\frac{1+(e^{\alpha_{L,t}}-1)B_{1,J}(u)
            +(e^{\alpha_{R,t}}-1)B_{J,J}(u)}{Z_{t,J}},
 \label{S-eq:endpoint-density-form}
\end{equation}
and its values at $u=0$ and $u=1$ are
\begin{equation}
 h_J(0\mid q_{t,J})=\frac{e^{\alpha_{L,t}}}{Z_{t,J}},\qquad
 h_J(1\mid q_{t,J})=\frac{e^{\alpha_{R,t}}}{Z_{t,J}}.
 \label{S-eq:endpoint-density-values}
\end{equation}
Hence
\begin{align}
 a_{{\rm odd},t}
  &=\frac12\{\log h_J(1\mid q_{t,J})
                -\log h_J(0\mid q_{t,J})\},
   \label{S-eq:odd-endpoint-identity}\\
 a_{{\rm even},t}
  &=\frac12\{\log h_J(0\mid q_{t,J})
                +\log h_J(1\mid q_{t,J})\}+\log Z_{t,J}.
   \label{S-eq:even-endpoint-identity}
\end{align}
\end{lemma}

\begin{proof}
Define the vector
\[
 c_{t,J}=(\alpha_{L,t},0,\ldots,0,\alpha_{R,t})'.
\]
Then $q_{t,J}=Q_Jc_{t,J}$. Since
\[
 Q_J'Q_J=I_J-\frac1J\mathbf 1_J\mathbf 1_J',
\]
the entries of $Q_J'q_{t,J}$ differ from those of $c_{t,J}$ by a common
constant. This constant cancels between the numerator and denominator of
the softmax formula in \eqref{eq:dcd-weights}. The unnormalized weights
are therefore $m_{1,J}e^{\alpha_{L,t}}$ and
$m_{J,J}e^{\alpha_{R,t}}$ for the first and last components, and $m_{j,J}$
for the interior components. Their sum is $Z_{t,J}$, which gives
\eqref{S-eq:endpoint-weights-J}. Thus every weight includes the same
normalizing factor.

Substituting these weights into the mixture density and using
$\sum_jB_{j,J}(u)=1$ gives \eqref{S-eq:endpoint-density-form}. The clamped
basis satisfies $B_{1,J}(0)=B_{J,J}(1)=1$.
At $u=0$, all basis functions except $B_{1,J}$ are zero. At $u=1$, all
except $B_{J,J}$ are zero. This gives
\eqref{S-eq:endpoint-density-values}. Subtracting and adding the logarithms
of these two values gives \eqref{S-eq:odd-endpoint-identity} and
\eqref{S-eq:even-endpoint-identity}.
\end{proof}

\subsection{Limit as J increases}

\begin{proof}[Proof of Proposition~\ref{thm:endpoint-intensity}]
Choose a finite $A$ such that
\[
 |a_{{\rm even},t}|+|a_{{\rm odd},t}|\le A
 \quad\text{for every }t.
\]
Then $|\alpha_{L,t}|,|\alpha_{R,t}|\le A$. With $C=e^A+1$,
\[
 |Z_{t,J}-1|\le C(m_{1,J}+m_{J,J}).
\]
The right side tends to zero and does not depend on $t$. Also,
$Z_{t,J}\ge e^{-A}$ because it is a weighted average of
$e^{\alpha_{L,t}}$, $e^{\alpha_{R,t}}$, and one. It follows that
\begin{equation}
 \sup_t|Z_{t,J}-1|\to0,\qquad
 \sup_t|\log Z_{t,J}|\to0.
 \label{S-eq:Z-to-one}
\end{equation}
We have
\[
 \left|\frac{e^{\alpha_{L,t}}}{Z_{t,J}}-e^{\alpha_{L,t}}\right|
 \le e^{2A}|Z_{t,J}-1|,
\]
with the same bound for $\alpha_{R,t}$. Hence
\begin{equation}
\begin{aligned}
 h_J(0\mid q_{t,J})
    &=e^{a_{{\rm even},t}-a_{{\rm odd},t}}+o(1),\\
 h_J(1\mid q_{t,J})
    &=e^{a_{{\rm even},t}+a_{{\rm odd},t}}+o(1),
\end{aligned}
 \label{S-eq:endpoint-intensity-limit}
\end{equation}
uniformly in $t$. For each fixed $J$, the reference CDF tends to zero and
one in the left and right tails. Since $h_J$ is continuous at $u=0$ and
$u=1$, the corrected return density satisfies
\begin{equation}
\begin{aligned}
 \lim_{r\to-\infty}\frac{\bar g_{t,J}(r)}{f_t(r)}
    &=h_J(0\mid q_{t,J}),\\
 \lim_{r\to+\infty}\frac{\bar g_{t,J}(r)}{f_t(r)}
    &=h_J(1\mid q_{t,J}).
\end{aligned}
 \label{S-eq:return-tail-ratios}
\end{equation}
Taking $J\to\infty$ gives \eqref{eq:return-tail-intensity-limit}.
For a Student-$t$ reference with degrees of freedom $\eta_t$, the density
decays as $|r|^{-(\eta_t+1)}$ in both tails. Multiplication by a positive
constant changes the tail coefficient and preserves this
exponent.
\end{proof}

\subsection{Integrated difference from the reference}

The values $h_J(0\mid q_{t,J})$ and $h_J(1\mid q_{t,J})$ need not tend to
one, but the integrated absolute difference between $h_J$ and the uniform
density does tend to zero. To see
this, subtract one from \eqref{S-eq:endpoint-density-form}:
\[
 h_J(u\mid q_{t,J})-1
  =\frac{(e^{\alpha_{L,t}}-1)\{B_{1,J}(u)-m_{1,J}\}
          +(e^{\alpha_{R,t}}-1)\{B_{J,J}(u)-m_{J,J}\}}{Z_{t,J}}.
\]
Because the basis functions are nonnegative,
\[
 \int_0^1|B_{1,J}(u)-m_{1,J}|\,du\le2m_{1,J},
\]
and the same bound holds with $B_{J,J}$ and $m_{J,J}$. Therefore,
\begin{align*}
 \int_0^1|h_J(u\mid q_{t,J})-1|\,du
 &\le\frac{2}{Z_{t,J}}
       \left\{|e^{\alpha_{L,t}}-1|m_{1,J}
               +|e^{\alpha_{R,t}}-1|m_{J,J}\right\}\\
 &\le2e^AC(m_{1,J}+m_{J,J}).
\end{align*}
The constants $A$ and $C$ are those used in the proof above. Thus
\[
 \sup_t\int_0^1|h_J(u\mid q_{t,J})-1|\,du
     =O(m_{1,J}+m_{J,J})\to0.
\]
Thus the values of $h_J$ at $u=0$ and $u=1$ can have limits different from
one even though $h_J$ converges to one in integrated absolute difference.

% Source: supp_theory.tex
\section{Detailed assumptions and proofs}
\label{S-supp:assumptions}

Throughout, $\|\cdot\|$ denotes the Euclidean norm for vectors and the
corresponding induced operator norm for matrices.  For a third-order tensor
$\mathcal T=(T_{rkl})$, $\|\mathcal T\|_F$ denotes the Frobenius norm,
\[
    \|\mathcal T\|_F
    :=
    \left(\sum_{r,k,l}T_{rkl}^2\right)^{1/2}.
\]

The reference distribution and Fisher projection matrices are known, and
the spline dimension $J$ is fixed. We prove consistency and asymptotic
normality of the six DCD-SM dynamic coefficient estimators in M4 and M5.
The weights $w_{t,i}$ are nonnegative and sum to one. Individual conditional
densities may differ, and returns within a date may be dependent.
Bounds on individual spline scores and their state derivatives also
bound their weighted averages.

\subsection{Consistency}
\label{S-subsec:consistency}

We establish consistency for the score update equation used in M4 and M5.
We assume that, near the true parameter, it has a unique stationary causal
solution and that the effect of initialization decays uniformly at an
exponential rate.

Use the weights $p(q)$, density $h(u\mid q)$, and responsibilities
$\tau(u,q)$ from Sections~\ref{subsec:spline-parameterization} and
\ref{subsec:description} of the main paper, with $QQ'=I_{J-1}$ and
$Q\mathbf 1_J=0$. For
$u_{t,i}=F_t(r_{t,i})$, the weighted spline score is
$s_t(q)=\sum_i w_{t,i}\psi(u_{t,i},q)$, where
\begin{equation}
    \psi(u,q)=Q\{\tau(u,q)-p(q)\}.
    \label{S-eq:analytic-score}
\end{equation}

\begin{lemma}
\label{S-lem:score-bounds}
For fixed $J$,
\[
    \sup_{u,q}\|\psi(u,q)\|\le\sqrt2,
    \qquad
    \sup_{u,q}
    \left\|
      \frac{\partial\psi(u,q)}{\partial q'}
    \right\|
    \le1.
\]
Let
\[
    \mathcal T_\psi(u,q)
    :=
    \left[
      \frac{\partial^2\psi_r(u,q)}
           {\partial q_k\,\partial q_l}
    \right]_{r,k,l=1}^{J-1}
\]
denote the third-order tensor of second derivatives of the vector score.
Then there exists a finite constant $C_J$, depending only on the fixed spline
dimension $J$, such that
\begin{equation}
    \sup_{u,q}\|\mathcal T_\psi(u,q)\|_F
    \le C_J<\infty.
    \label{S-eq:score-second-derivative-bound}
\end{equation}
Therefore, a first-order Taylor expansion in $q$ gives, for any
$q,\widetilde q\in\mathbb R^{J-1}$,
\begin{equation}
    \|\psi(u,q)-\psi(u,\widetilde q)\|
    \le
    \|q-\widetilde q\|,
    \label{S-eq:score-taylor-bound}
\end{equation}
and
\begin{equation}
    \left\|
      \frac{\partial\psi(u,q)}{\partial q'}
      -
      \frac{\partial\psi(u,\widetilde q)}{\partial q'}
    \right\|
    \le
    C_J\|q-\widetilde q\|.
    \label{S-eq:score-jacobian-taylor-bound}
\end{equation}
\end{lemma}

\begin{proof}
Both $\tau(u,q)$ and $p(q)$ are probability vectors, so
$\|\tau(u,q)-p(q)\|\le\sqrt2$. Since $\|Q\|=1$, the score bound follows.
Writing $V(a)=\operatorname{diag}(a)-aa'$, differentiation gives
\[
    \frac{\partial\psi(u,q)}{\partial q'}
    =Q\{V(\tau(u,q))-V(p(q))\}Q'.
\]
For a probability vector $a$, the diagonal entry of $V(a)$ and the sum of
the absolute off-diagonal entries in the same row both equal
$a_i(1-a_i)$. Gershgorin's theorem gives
$\|V(a)\|\le\max_i2a_i(1-a_i)\le1/2$, so the score Jacobian has norm at
most one.

For a weighted softmax vector
$a_j(z)=c_je^{z_j}/\sum_s c_se^{z_s}$ with $c_j\ge0$ and $\sum_jc_j>0$,
\begin{equation}
    \frac{\partial a_j(z)}{\partial z_l}
    =a_j(z)\{\delta_{jl}-a_l(z)\}.
    \label{S-eq:softmax-first-derivative}
\end{equation}
Consequently,
\begin{equation}
    \frac{\partial V_{jk}\{a(z)\}}{\partial z_l}
    =a_j(\delta_{jl}-a_l)(\delta_{jk}-a_k)
     -a_ja_k(\delta_{kl}-a_l).
    \label{S-eq:softmax-covariance-derivative}
\end{equation}
The right side is uniformly bounded. Both $p(q)$ and $\tau(u,q)$ are
weighted softmax vectors with index $Q'q$. Applying the chain rule to the
score Jacobian therefore gives a uniform bound on every entry of
$\mathcal T_\psi(u,q)$. Since $J$ and $Q$ are fixed, its Frobenius norm is
bounded by a finite constant $C_J$. The first-order Taylor formula gives
\eqref{S-eq:score-taylor-bound} and \eqref{S-eq:score-jacobian-taylor-bound}.
\end{proof}

\subsubsection*{DCD-SM as a stochastic recurrence equation}

Let $E=(e_{\rm even},e_{\rm odd})$ contain the vectors in
\eqref{eq:endpoint-directions}. Define
\begin{equation}
    z_t=(a_{{\rm even},t},a_{{\rm odd},t})',\qquad q_t=Ez_t.
    \label{S-eq:endpoint-z-state}
\end{equation}

For M4 and M5, write the fixed-anchor Fisher projection as
\[
    s_t^\perp(q)=s_t(q)+b_t,
    \qquad
    b_t
    :=
    -I_{q,ms,0}I_{ms,ms,0}^{-1}s_{ms,t},
    \qquad
    ms=(m,\log\sigma^2).
\]
The projection coefficient $I_{q,ms,0}I_{ms,ms,0}^{-1}$ is computed once at
$(ms,q)=(ms_0,0)$ and held fixed. The reference score $s_{ms,t}$ is computed
from the date-$t$ returns using the known reference distribution, as in
\eqref{eq:reference-ms-score}. Define
\begin{equation}
    r_t(z)
    :=
    E's_t^\perp(Ez)
    =
    E'\{s_t(Ez)+b_t\}
    \in\mathbb R^2.
    \label{S-eq:endpoint-score-r}
\end{equation}
With the observed returns and reference distribution held fixed, its state
Jacobian is
\begin{equation}
    G_t(z)
    :=
    \frac{\partial r_t(z)}{\partial z'}
    =
    E'H_t(Ez)E,
    \qquad
    H_t(q)
    :=
    \frac{\partial s_t(q)}{\partial q'}.
    \label{S-eq:endpoint-score-jacobian}
\end{equation}
Lemma~\ref{S-lem:score-bounds} implies that $G_t(z)$ is
globally bounded and globally Lipschitz in $z$.

The parameter vector contains the six HAR coefficients,
\begin{equation}
    \theta=\kappa=
    (\kappa_{{\rm even},1},\kappa_{{\rm even},5},\kappa_{{\rm even},22},
     \kappa_{{\rm odd},1},\kappa_{{\rm odd},5},\kappa_{{\rm odd},22})',
    \qquad d_\theta=6.
    \label{S-eq:downstream-parameter}
\end{equation}
For $h\in\{1,5,22\}$, define
\begin{equation}
    C_h(\theta)
    :=
    \operatorname{diag}
    \left(
      \kappa_{{\rm even},h},
      \kappa_{{\rm odd},h}
    \right).
    \label{S-eq:har-coefficient-matrix}
\end{equation}
Expanding the HAR$(1,5,22)$ moving averages in
\eqref{eq:parity-har-dcd} gives
\begin{equation}
    z_{t+1}
    =
    \sum_{j=0}^{21}
    K_j(\theta)r_{t-j}(z_{t-j}),
    \label{S-eq:endpoint-expanded-har}
\end{equation}
where
\[
 K_j(\theta)=
 \begin{cases}
 C_1(\theta)+\frac15 C_5(\theta)+\frac1{22}C_{22}(\theta),&j=0,\\
 \frac15 C_5(\theta)+\frac1{22}C_{22}(\theta),&j=1,\ldots,4,\\
 \frac1{22}C_{22}(\theta),&j=5,\ldots,21.
 \end{cases}
\]

Now define the companion state
\begin{equation}
    X_t
    :=
    \begin{pmatrix}
      z_t'&z_{t-1}'&\cdots&z_{t-21}'
    \end{pmatrix}'
    \in\mathbb R^{44}.
    \label{S-eq:companion-state}
\end{equation}
Then \eqref{S-eq:endpoint-expanded-har} is the stochastic recurrence equation
\begin{equation}
    X_{t+1}
    =
    \Phi_t(X_t;\theta).
    \label{S-eq:endpoint-sre}
\end{equation}
For
$x=(z_t',z_{t-1}',\ldots,z_{t-21}')'$, let
\[
    A_t(x;\theta)
    :=
    \frac{\partial\Phi_t(x;\theta)}{\partial x'}.
\]
Its companion form is
\begin{equation}
A_t(x;\theta)
=
\begin{pmatrix}
K_0(\theta)G_t(z_t)
&
K_1(\theta)G_{t-1}(z_{t-1})
&
\cdots
&
K_{21}(\theta)G_{t-21}(z_{t-21})
\\
I_2&0&\cdots&0\\
0&I_2&\ddots&\vdots\\
\vdots&&\ddots&0
\end{pmatrix}.
\label{S-eq:companion-jacobian}
\end{equation}
Thus local stability of the score update equation is governed by products of the feedback
matrices $K_j(\theta)G_{t-j}$ rather than by the HAR lag matrices alone.

\subsubsection*{Assumptions}

\paragraph{Assumption 1 (Conditional weighted-average-density specification).}
The reference distribution $F_t$, the constituent set and its finite size
$N_t\ge1$, and the weights $w_{t,i}\ge0$ with $\sum_iw_{t,i}=1$ are given
before the date-$t$ returns are observed. They are
$\mathcal F_{t-1}$-measurable and do not depend on $\theta$. The Fisher
projection matrices are known and held fixed. Conditional on $\mathcal F_{t-1}$, each
$u_{t,i}$ has density $\phi_{t,i}$ on $[0,1]$, and
\begin{equation}
    \bar\phi_t(u)
    :=\sum_{i=1}^{N_t}w_{t,i}\phi_{t,i}(u)
    =h\{u\mid q_t(\theta_0)\}
    \quad\text{Lebesgue-a.e., a.s.}
    \label{S-eq:oracle-average-density}
\end{equation}
The individual densities need not be equal, and conditional independence
across $i$ is not required.  No stationarity or ergodicity assumption is
imposed here on the full cross-sectional observation process.

For later use, Assumption~1 implies that any jointly measurable random
function $a_t(u)$, measurable with respect to $\mathcal F_{t-1}$ for each
$u$ and satisfying
$E[\int_0^1|a_t(u)|\bar\phi_t(u)\,du]<\infty$, obeys
\begin{equation}
\begin{split}
    E\left[\sum_{i=1}^{N_t}w_{t,i}a_t(u_{t,i})
      \,\middle|\,\mathcal F_{t-1}\right]
    &=\sum_{i=1}^{N_t}w_{t,i}
       \int_0^1a_t(u)\phi_{t,i}(u)\,du\\
    &=\int_0^1a_t(u)h\{u\mid q_t(\theta_0)\}\,du.
\end{split}
\label{S-eq:conditional-average-identity}
\end{equation}
The identity also applies componentwise to vector- and matrix-valued
functions.

\paragraph{Assumption 2 (Stationarity of $q_t$ and local stability).}
The true parameter $\theta_0$ belongs to the interior of a compact
parameter neighborhood $\mathcal N$.  For every $\theta\in\mathcal N$, the
recursion \eqref{S-eq:endpoint-sre} admits a unique causal solution
$X_t(\theta)$, measurable with respect to $\mathcal F_{t-1}$.  The associated
spline-state process $q_t(\theta)$ is almost surely continuous in $\theta$
on $\mathcal N$, and
\[
    \{q_t(\cdot)\}_{t\in\mathbb Z}
    \quad\text{in }C(\mathcal N;\mathbb R^{J-1})
\]
is strictly stationary and ergodic.  This is a joint condition on the
candidate state family, including $q_t(\theta_0)$, rather than a separate
condition for each parameter value.  Since $E$ has full column rank,
$z_t(\theta)=(E'E)^{-1}E'q_t(\theta)$. The companion state $X_t(\theta)$
stacks finitely many lags of $z_t(\theta)$.  Thus $X_t(\cdot)$ inherits
continuity, strict stationarity, and ergodicity. There is a finite constant
$C_b$ such that $\|b_t\|\le C_b$ almost surely for all $t$.
Moreover, for any fixed initialization in a common local basin of attraction,
there exist a finite random variable $C$ and a constant $\rho\in(0,1)$ such
that
\begin{equation}
    \sup_{\theta\in\mathcal N}
    \|\widehat X_t(\theta)-X_t(\theta)\|
    \le
    C\rho^t
    \qquad\text{a.s. for all sufficiently large }t.
    \label{S-eq:local-uniform-forgetting}
\end{equation}

\paragraph{Assumption 3 (Local dynamic identification).}
For the given reference distribution and Fisher projection matrices, and for
$\theta\in\mathcal N$,
\[
    q_t(\theta)=q_t(\theta_0)
    \quad\text{a.s. for all }t
\]
implies $\theta=\theta_0$.

\subsubsection*{Uniform convergence of the local criterion}

For $\theta\in\mathcal N$, define
\[
    \ell_t(\theta)
    =
    \sum_{i=1}^{N_t}w_{t,i}
      \log h\{u_{t,i}\mid q_t(\theta)\},
    \qquad
    L_T(\theta)
    =
    \frac1T\sum_{t=1}^T\ell_t(\theta),
\]
where $q_t(\theta)$ is obtained from the first two entries of $X_t(\theta)$
using \eqref{S-eq:endpoint-z-state}. Let
\[
    L(\theta)=E[\ell_t(\theta)].
\]
Because the reference distribution at each date does not depend on $\theta$,
maximizing $L_T(\theta)$ is equivalent to maximizing the
weighted cross-sectional return log-score criterion over $\mathcal N$.  Define the local
estimator by
\begin{equation}
    \widehat\theta_T
    \in
    \arg\max_{\theta\in\mathcal N}L_T(\theta).
    \label{S-eq:local-estimator}
\end{equation}

The criterion is uniformly bounded. Indeed, by
Lemma~\ref{S-lem:score-bounds}, $s_t(q)$ is uniformly bounded, and $b_t$ is
uniformly bounded by Assumption~2.  Since the matrices $K_j(\theta)$ are
bounded on the compact neighborhood $\mathcal N$,
\eqref{S-eq:endpoint-expanded-har} implies
\[
    \sup_{t,\theta\in\mathcal N}\|z_t(\theta)\|
    \le M_z<\infty,
\]
and hence
\[
    \sup_{t,\theta\in\mathcal N}\|q_t(\theta)\|
    \le M_q<\infty.
\]
If $v=Q'q$, then
\[
    h(u\mid q)
    =
    \frac{
      \sum_{j=1}^J B_j(u)e^{v_j}
    }{
      \sum_{j=1}^J m_je^{v_j}
    }.
\]
Using the partition-of-unity identities
$\sum_jB_j(u)=1$ and $\sum_jm_j=1$ gives
\begin{equation}
    e^{-2M_q}
    \le
    h\{u\mid q_t(\theta)\}
    \le
    e^{2M_q}
    \label{S-eq:local-density-envelope}
\end{equation}
uniformly in $u,t$, and $\theta\in\mathcal N$.

\begin{lemma}
\label{S-lem:local-ulln}
Under Assumptions~1--2,
\begin{equation}
    \sup_{\theta\in\mathcal N}
    |L_T(\theta)-L(\theta)|
    \longrightarrow0
    \qquad\text{a.s.}
    \label{S-eq:local-uniform-lln}
\end{equation}
\end{lemma}

\begin{proof}
For each fixed $\theta\in\mathcal N$, decompose the criterion as
\[
    \ell_t(\theta)=d_t(\theta)+m_t(\theta),
    \qquad
    m_t(\theta):=E\{\ell_t(\theta)\mid\mathcal F_{t-1}\}.
\]
The conditional-average identity gives
\begin{equation}
    m_t(\theta)
    =\int_0^1 h\{u\mid q_t(\theta_0)\}
       \log h\{u\mid q_t(\theta)\}\,du.
    \label{S-eq:predictable-criterion}
\end{equation}
Thus $m_t(\theta)$ is a measurable function of the jointly stationary and
ergodic family $q_t(\theta)$ in Assumption~2. Its expectation
$E[m_t(\theta)]=L(\theta)$ does not depend on $t$. The bound
\eqref{S-eq:local-density-envelope} implies
$|\ell_t(\theta)|,|m_t(\theta)|\le2M_q$, so $d_t(\theta)$ is a bounded
martingale difference.  In particular,
\[
    \sum_{t=1}^{\infty}\frac{E[d_t(\theta)^2]}{t^2}<\infty,
    \qquad
    \frac1T\sum_{t=1}^T d_t(\theta)\longrightarrow0
    \quad\text{a.s.}
\]
The second statement follows from martingale convergence and Kronecker's
lemma.  Applying the ergodic theorem to $m_t(\theta)$ now gives
\[
    L_T(\theta)\to L(\theta)
    \qquad\text{a.s.}
\]
for every fixed $\theta$.

To obtain uniform convergence without assuming stationarity of
$\ell_t(\theta)$, define
\[
    \omega_t(\delta)
    :=
    \sup_{\substack{
      \theta,\widetilde\theta\in\mathcal N\\
      \|\theta-\widetilde\theta\|\le\delta
    }}
    \|q_t(\theta)-q_t(\widetilde\theta)\|.
\]
Assumption~2 implies that $\omega_t(\delta)\to0$ almost surely as
$\delta\downarrow0$, while $\omega_t(\delta)\le2M_q$.  The bounded score
in Lemma~\ref{S-lem:score-bounds} gives the bound
\[
    |\ell_t(\theta)-\ell_t(\widetilde\theta)|
    \le\sqrt2\,\omega_t(\delta)
    \quad\text{if }\|\theta-\widetilde\theta\|\le\delta.
\]
The same bound holds for $|m_t(\theta)-m_t(\widetilde\theta)|$, since
$\omega_t(\delta)$ is predictable.  Dominated convergence gives
\[
    E[\omega_t(\delta)]\to0.
\]
For fixed $\delta$, $\omega_t(\delta)$ is stationary and ergodic by
Assumption~2, so
\[
    \frac1T\sum_{t=1}^T\omega_t(\delta)
    \to
    E[\omega_t(\delta)]
    \qquad\text{a.s.}
\]
A finite $\delta$-net of the compact set $\mathcal N$, followed by
$\delta\downarrow0$, proves \eqref{S-eq:local-uniform-lln}.  Only the
conditional expectation $m_t(\theta)$ and $\omega_t(\delta)$ require an ergodic theorem.
The observation-dependent remainder is handled by a martingale law of
large numbers.
\end{proof}

\subsubsection*{Consistency theorem}

\begin{theorem}
\label{S-thm:consistency}
Suppose that Assumptions~1--3 hold.  Then the local estimator in
\eqref{S-eq:local-estimator} satisfies
\[
    \widehat\theta_T
    \longrightarrow
    \theta_0
    \qquad\text{a.s.}
\]
The same conclusion holds for the criterion evaluated from any fixed
initialization covered by Assumption~2.
\end{theorem}

\begin{proof}
Lemma~\ref{S-lem:local-ulln} gives uniform convergence of the criterion on
$\mathcal N$.  It remains to identify its population maximizer.

Since $q_t(\theta)$ is $\mathcal F_{t-1}$-measurable,
\eqref{S-eq:conditional-average-identity} applied to the log density ratio gives
\begin{align}
L(\theta_0)-L(\theta)
&=
E\left[
  \sum_{i=1}^{N_t}w_{t,i}
  \log
  \frac{
    h\{u_{t,i}\mid q_t(\theta_0)\}
  }{
    h\{u_{t,i}\mid q_t(\theta)\}
  }
\right]
\nonumber\\
&=
E\left[
  \int_0^1
  h\{u\mid q_t(\theta_0)\}
  \log
  \frac{
    h\{u\mid q_t(\theta_0)\}
  }{
    h\{u\mid q_t(\theta)\}
  }
  \,du
\right]
\nonumber\\
&=
E\left[
  {\rm KL}
  \left\{
    h(\cdot\mid q_t(\theta_0))
    \,\Vert\,
    h(\cdot\mid q_t(\theta))
  \right\}
\right]
\ge0.
\label{S-eq:population-kl-gap}
\end{align}
If equality holds, the nonnegative KL divergence is zero almost surely, so
\[
    h(\cdot\mid q_t(\theta))
    =
    h(\cdot\mid q_t(\theta_0))
    \quad\text{Lebesgue-a.e., a.s.}
\]
Linear independence of the normalized B-spline basis implies
\[
    p\{q_t(\theta)\}
    =
    p\{q_t(\theta_0)\}
    \quad\text{a.s.}
\]
The weighted-softmax mapping is one-to-one because
\[
    Q\log\{p(q)/m\}
    =
    QQ'q
    -
    \log\!\left(
      \sum_{j=1}^Jm_je^{(Q'q)_j}
    \right)
    Q\mathbf1_J
    =q.
\]
Hence
\[
    q_t(\theta)=q_t(\theta_0)
    \quad\text{a.s. for all }t,
\]
and Assumption~3 yields $\theta=\theta_0$.  Thus $L(\theta)$ has the unique
maximizer $\theta_0$ on $\mathcal N$.  Combining identification with
\eqref{S-eq:local-uniform-lln} gives
\[
    \widehat\theta_T\to\theta_0
    \qquad\text{a.s.}
\]
by the usual argmax argument.

Finally, local-uniform forgetting in
\eqref{S-eq:local-uniform-forgetting}, together with the bounded score, implies
that computing the score update equation from fixed initial values in the
same local basin changes the average criterion by $o(1)$ uniformly on
$\mathcal N$.  The initialized estimator therefore has the same consistency
limit.
\end{proof}

\subsection{Asymptotic Normality}
\label{S-subsec:asymptotic-normality}

We now derive the asymptotic distribution of the local estimator.
The derivatives of $X_t(\theta)$ are obtained by differentiating the score
update equation. Their existence, continuity, and moments are stated
directly below.

Let
\[
    S
    :=
    \begin{pmatrix}
        I_2 & 0 & \cdots & 0
    \end{pmatrix}
    \in\mathbb R^{2\times44},
\]
so that $z_t(\theta)=SX_t(\theta)$ and
\[
    q_t(\theta)=ESX_t(\theta).
\]
Take the compact neighborhood $\mathcal N$ in Assumption~2 to be convex,
shrinking it if necessary. The true parameter $\theta_0$ remains an interior point.

\subsubsection*{Derivative stochastic recursions}

Define the first parameter derivative matrix
\begin{equation}
    D_t(\theta)
    :=
    \frac{\partial X_t(\theta)}{\partial\theta'}
    \in\mathbb R^{44\times d_\theta}.
    \label{S-eq:first-derivative-state}
\end{equation}
Differentiating the SRE \eqref{S-eq:endpoint-sre} by the ordinary chain rule
gives
\begin{equation}
    D_{t+1}(\theta)
    =
    A_t\{X_t(\theta);\theta\}D_t(\theta)
    +
    C_t(\theta),
    \label{S-eq:first-derivative-sre}
\end{equation}
where
\[
    C_t(\theta)
    :=
    \left.
      \frac{\partial\Phi_t(x;\theta)}{\partial\theta'}
    \right|_{x=X_t(\theta)}
    \in\mathbb R^{44\times d_\theta}.
\]
In the empirical HAR specification the matrices $K_j(\theta)$ are
linear in the six HAR coefficients.  Because $r_t(z)$ is uniformly bounded,
$C_t(\theta)$ is locally uniformly bounded.

For the second derivative, define the third-order tensor
\begin{equation}
    R_t(\theta)
    :=
    \left[
      \frac{\partial^2 X_{t,r}(\theta)}
           {\partial\theta_j\,\partial\theta_k}
    \right]_
    {\substack{r=1,\ldots,44\\j,k=1,\ldots,d_\theta}}.
    \label{S-eq:second-derivative-state}
\end{equation}
Write $A_{t,ra}$ and $D_{t,aj}$ for the entries of
$A_t\{X_t(\theta);\theta\}$ and $D_t(\theta)$, respectively.  A second
application of the chain rule gives, for
$r=1,\ldots,44$ and $j,k=1,\ldots,d_\theta$,
\begin{align}
    R_{t+1,rjk}
    ={}&
    \sum_{a=1}^{44}A_{t,ra}R_{t,ajk}
    \nonumber\\
    &+
    \sum_{a=1}^{44}\sum_{b=1}^{44}
    \frac{\partial^2\Phi_{t,r}}
         {\partial x_a\,\partial x_b}
    D_{t,aj}D_{t,bk}
    \nonumber\\
    &+
    \sum_{a=1}^{44}
    \frac{\partial^2\Phi_{t,r}}
         {\partial x_a\,\partial\theta_k}
    D_{t,aj}
    +
    \sum_{a=1}^{44}
    \frac{\partial^2\Phi_{t,r}}
         {\partial x_a\,\partial\theta_j}
    D_{t,ak}
    \nonumber\\
    &+
    \frac{\partial^2\Phi_{t,r}}
         {\partial\theta_j\,\partial\theta_k}.
    \label{S-eq:second-derivative-sre}
\end{align}
All derivatives of $\Phi_t$ in \eqref{S-eq:second-derivative-sre} are evaluated
at $\{X_t(\theta),\theta\}$.  For the empirical HAR parameterization,
\[
    \frac{\partial^2\Phi_{t,r}}
         {\partial\theta_j\,\partial\theta_k}
    =0.
\]
Moreover, Lemma~\ref{S-lem:score-bounds} implies that the second derivatives of
$\Phi_t$ with respect to the state and the mixed state--parameter derivatives
are locally uniformly bounded for $X_t(\theta)$, $\theta\in\mathcal N$. Hence the nonhomogeneous
part of \eqref{S-eq:second-derivative-sre} is bounded in norm by a constant
multiple of
\[
    1+\|D_t(\theta)\|+\|D_t(\theta)\|^2.
\]

\paragraph{Assumption 4 (Derivative regularity and moments).}
On $\mathcal N$, $X_t(\theta)$ is twice continuously
differentiable in $\theta$ almost surely.  The derivative processes
$D_t(\theta)$ and $R_t(\theta)$ defined above are causal and almost surely
continuous in $\theta$.  The random-function process
$\{(X_t(\cdot),D_t(\cdot),R_t(\cdot))\}$ is jointly strictly stationary
and ergodic in the corresponding product of continuous-function spaces.
For some
$\delta>0$,
\begin{equation}
    E\left[
      \sup_{\theta\in\mathcal N}
      \|D_t(\theta)\|^{2+\delta}
    \right]
    <\infty,
    \qquad
    E\left[
      \sup_{\theta\in\mathcal N}
      \|R_t(\theta)\|_F^{1+\delta}
    \right]
    <\infty.
    \label{S-eq:derivative-moments}
\end{equation}
For any fixed initialization covered by Assumption~2, the corresponding first-
and second derivatives computed from these initial values satisfy
\begin{equation}
    \sup_{\theta\in\mathcal N}
    \left\{
      \|\widehat D_t(\theta)-D_t(\theta)\|
      +
      \|\widehat R_t(\theta)-R_t(\theta)\|_F
    \right\}
    \le C_D\rho_D^t
    \quad\text{a.s.}
    \label{S-eq:derivative-forgetting}
\end{equation}
for some finite random variable $C_D$ and some $\rho_D\in(0,1)$.

Define
\begin{equation}
    J_t(\theta)
    :=
    \frac{\partial q_t(\theta)}{\partial\theta'}
    =
    ES
    \frac{\partial X_t(\theta)}{\partial\theta'}
    \in\mathbb R^{(J-1)\times d_\theta},
    \label{S-eq:q-parameter-derivative}
\end{equation}
and write $J_t=J_t(\theta_0)$.

Let
\[
    \mathcal I_q(q)
    :=
    \int_0^1
        \psi(u,q)\psi(u,q)'
        h(u\mid q)\,du.
\]

\paragraph{Assumption 5 (Local information and variance stabilization).}
The negative expected Hessian
\begin{equation}
    \mathcal H
    :=
    -E[\nabla_\theta^2\ell_t(\theta_0)]
    =E\left[
        J_t'
        \mathcal I_q\{q_t(\theta_0)\}
        J_t
    \right]
    \label{S-eq:sensitivity-matrix}
\end{equation}
is positive definite.
There is a finite nonrandom matrix $\Omega$ such that
\begin{equation}
    \frac1T\sum_{t=1}^T J_t'\Sigma_tJ_t
    \longrightarrow\Omega
    \qquad\text{in probability},
    \label{S-eq:predictable-variance-limit}
\end{equation}
where
\[
    \Sigma_t
    :=\operatorname{Var}\{s_t(q_t(\theta_0))\mid\mathcal F_{t-1}\}.
\]
This condition is used for asymptotic normality, not consistency.  It
controls the aggregate effect of within-date heterogeneity and dependence
without requiring stationarity or ergodicity of the full observation
process.

\subsubsection*{Martingale score and Hessian}

The parameter score at the truth is
\begin{equation}
    S_t
    :=
    \nabla_\theta\ell_t(\theta_0)
    =
    J_t'
    s_t\{q_t(\theta_0)\}.
    \label{S-eq:parameter-score}
\end{equation}
Its conditional covariance is
\begin{equation}
    E[S_tS_t'\mid\mathcal F_{t-1}]
    =
    J_t'\Sigma_tJ_t,
    \label{S-eq:Omega-matrix}
\end{equation}
whose time average converges to $\Omega$ by Assumption~5.  Assumption~4
and the bounded DCD-SM score imply
$\sup_t E\|S_t\|^{2+\delta}<\infty$.

\begin{lemma}
\label{S-lem:score-clt}
Under Assumptions~1, 2, 4, and 5,
\begin{equation}
    \frac1{\sqrt T}\sum_{t=1}^TS_t
    \Longrightarrow
    N(0,\Omega).
    \label{S-eq:score-martingale-clt}
\end{equation}
\end{lemma}

\begin{proof}
The solution $X_t(\theta_0)$ and its derivative $D_t$ are measurable
with respect to $\mathcal F_{t-1}$.  Hence $J_t$ is
$\mathcal F_{t-1}$-measurable.  Individual scores may have nonzero
conditional means.  Applying \eqref{S-eq:conditional-average-identity} to
$\psi\{u,q_t(\theta_0)\}$ instead gives
\begin{align*}
    E\left[
      s_t\{q_t(\theta_0)\}
      \mid\mathcal F_{t-1}
    \right]
    &=\int_0^1\psi\{u,q_t(\theta_0)\}
       \bar\phi_t(u)\,du\\
    &=\int_0^1\psi\{u,q_t(\theta_0)\}
       h\{u\mid q_t(\theta_0)\}\,du
      =0.
\end{align*}
The final equality is the density-score identity for the weighted-average-density
model.  Neither identical conditional marginals nor conditional independence
across $i$ is used in this argument.
Therefore
\[
    E(S_t\mid\mathcal F_{t-1})=0,
\]
and $\{S_t,\mathcal F_t\}$ is a martingale difference sequence.  Its
conditional covariance is
\[
    E(S_tS_t'\mid\mathcal F_{t-1})
    =
    J_t'\Sigma_tJ_t.
\]
By the variance-stabilization condition in Assumption~5,
\[
    \frac1T\sum_{t=1}^T
    E(S_tS_t'\mid\mathcal F_{t-1})
    \longrightarrow
    \Omega
    \qquad\text{in probability}.
\]
The uniform $(2+\delta)$ moment bound implies the conditional Lindeberg
condition in probability: for every $\epsilon>0$, the expectation of
\[
    \frac1T\sum_{t=1}^T
    E\left[\|S_t\|^2\mathbf1\{\|S_t\|>\epsilon\sqrt T\}
      \,\middle|\,\mathcal F_{t-1}\right]
\]
is $O(T^{-\delta/2})$.  The martingale central limit
theorem gives \eqref{S-eq:score-martingale-clt}.
\end{proof}

For $\theta\in\mathcal N$, define
\[
    H_t\{q_t(\theta)\}
    :=
    \frac{\partial s_t\{q_t(\theta)\}}{\partial q'}
    =
    \sum_{i=1}^{N_t}w_{t,i}
      \frac{\partial\psi\{u_{t,i},q_t(\theta)\}}{\partial q'}.
\]
Let
\[
    R_{q,t}(\theta)
    :=
    \left[
      \frac{\partial^2 q_{t,r}(\theta)}
           {\partial\theta_j\,\partial\theta_k}
    \right]_
    {\substack{r=1,\ldots,J-1\\j,k=1,\ldots,d_\theta}}
\]
denote the third-order tensor of second parameter derivatives of the embedded
spline state.  Since $q_t(\theta)=ESX_t(\theta)$,
\[
    R_{q,t,rjk}(\theta)
    =
    \sum_{a=1}^{44}
      (ES)_{ra}R_{t,ajk}(\theta).
\]

Write $J_{t,\cdot j}(\theta)$ for the $j$th column of
$J_t(\theta)$.  The $(j,k)$ element of the criterion Hessian is
\begin{equation}
\begin{split}
    \frac{\partial^2\ell_t(\theta)}
         {\partial\theta_j\,\partial\theta_k}
    ={}&
    J_{t,\cdot j}(\theta)'
    H_t\{q_t(\theta)\}
    J_{t,\cdot k}(\theta)\\
    &+
    \sum_{r=1}^{J-1}
      s_{t,r}\{q_t(\theta)\}
      R_{q,t,rjk}(\theta).
\end{split}
\label{S-eq:likelihood-hessian}
\end{equation}

At the truth, the conditional expectation of the second term is zero
because $R_{q,t}(\theta_0)$ is predictable and the weighted cross-sectional score has
conditional mean zero.  Applying \eqref{S-eq:conditional-average-identity}
to the score Jacobian and then using the density information identity gives
\begin{align*}
    E\left[
      H_t\{q_t(\theta_0)\}
      \mid\mathcal F_{t-1}
    \right]
    &=\int_0^1
      \frac{\partial\psi\{u,q_t(\theta_0)\}}{\partial q'}
      \bar\phi_t(u)\,du\\
    &=\int_0^1
      \frac{\partial\psi\{u,q_t(\theta_0)\}}{\partial q'}
      h\{u\mid q_t(\theta_0)\}\,du\\
    &=-\mathcal I_q\{q_t(\theta_0)\}.
\end{align*}
This is an identity for the expected average Hessian, not an information
equality for the actual weighted cross-sectional score covariance.
Consequently,
\begin{equation}
    -E\{\nabla_\theta^2\ell_t(\theta_0)\}
    =
    \mathcal H.
    \label{S-eq:population-hessian}
\end{equation}

\begin{lemma}
\label{S-lem:hessian-ulln}
Under Assumptions~1, 2, and 4,
\begin{equation}
    \sup_{\theta\in\mathcal N}
    \left\|
      \frac1T\sum_{t=1}^T
      \nabla_\theta^2\ell_t(\theta)
      -
      E\{\nabla_\theta^2\ell_t(\theta)\}
    \right\|
    \longrightarrow0
    \qquad\text{a.s.}
    \label{S-eq:hessian-ulln}
\end{equation}
The map
$\theta\mapsto E\{\nabla_\theta^2\ell_t(\theta)\}$
is continuous at $\theta_0$.
\end{lemma}

\begin{proof}
Lemma~\ref{S-lem:score-bounds} gives
\[
    \sup_{t,q}\|H_t(q)\|\le1
\]
and a uniform bound on the third-order tensor of first derivatives of
$H_t(q)$ with respect to $q$.  Therefore the first term of
\eqref{S-eq:likelihood-hessian} is bounded by a constant multiple of
$\|D_t(\theta)\|^2$, while the second is bounded by a constant multiple of
$\|R_t(\theta)\|_F$. The moment bounds in Assumption~4 apply to both
terms uniformly over $\theta\in\mathcal N$.

To separate the observations from $q_t(\theta)$ and its derivatives, let
$\Xi_t(u,\theta)$ be the one-observation Hessian contribution, with entries
\[
    \Xi_{t,jk}(u,\theta)
    =J_{t,\cdot j}(\theta)'
      \frac{\partial\psi\{u,q_t(\theta)\}}{\partial q'}
      J_{t,\cdot k}(\theta)
      +\sum_{r=1}^{J-1}\psi_r\{u,q_t(\theta)\}
        R_{q,t,rjk}(\theta).
\]
Then
\begin{align*}
    \nabla_\theta^2\ell_t(\theta)
    &=\sum_i w_{t,i}\Xi_t(u_{t,i},\theta),\\
    A_t^H(\theta)
    &:=E\{\nabla_\theta^2\ell_t(\theta)\mid\mathcal F_{t-1}\}\\
    &=\int_0^1\Xi_t(u,\theta)h\{u\mid q_t(\theta_0)\}\,du.
\end{align*}
By Assumption~4, $A_t^H(\theta)$ is stationary and ergodic for each fixed
$\theta$.  Moreover,
\[
    \sup_{u,\theta}\|\Xi_t(u,\theta)\|
    \le K_t
    :=C\sup_{\theta\in\mathcal N}
       \{\|D_t(\theta)\|^2+\|R_t(\theta)\|_F\},
\]
where $K_t$ is stationary and ergodic and has a finite $p$th moment for
some $1<p\le2$, by \eqref{S-eq:derivative-moments}.  Consequently the
centered Hessian is a martingale difference with uniformly bounded $p$th
moments.  Its time average converges to zero almost surely by the
martingale strong law, since $\sum_t t^{-p}<\infty$.  The ergodic theorem
applied to $A_t^H(\theta)$ yields pointwise convergence of the average
Hessian to its expectation, which is independent of $t$.

For uniformity, use the predictable modulus
\[
    \omega_t^H(\varepsilon)
    :=\sup_{\substack{\theta,\widetilde\theta\in\mathcal N\\
                      \|\theta-\widetilde\theta\|\le\varepsilon}}
       \sup_{u\in[0,1]}
       \|\Xi_t(u,\theta)-\Xi_t(u,\widetilde\theta)\|.
\]
The score bounds and Assumption~4 imply almost-sure convergence of this
modulus to zero as $\varepsilon\downarrow0$, uniformly in $u$.  It is a
stationary and ergodic function of $q_t(\theta)$ and its derivative families and
is bounded by $2K_t$.  It bounds both the observed Hessian modulus and
that of its conditional expectation.  Dominated convergence, the ergodic
theorem for this modulus, and the finite-net argument in
Lemma~\ref{S-lem:local-ulln} prove \eqref{S-eq:hessian-ulln}. The bound $2K_t$
gives continuity of the expectation by dominated convergence.
\end{proof}

\subsubsection*{Asymptotic-normality theorem}

\begin{theorem}
\label{S-thm:asymptotic-normality}
Suppose that Assumptions~1--5 hold.  Then the locally stable estimator in
\eqref{S-eq:local-estimator} satisfies
\begin{equation}
    \sqrt T(\widehat\theta_T-\theta_0)
    \Longrightarrow
    N\left(
      0,\,
      \mathcal H^{-1}\Omega\mathcal H^{-1}
    \right).
    \label{S-eq:asymptotic-normality}
\end{equation}
If $\Omega$ is positive definite, the limiting covariance matrix is positive
definite.
\end{theorem}

\begin{proof}
By Theorem~\ref{S-thm:consistency},
$\widehat\theta_T\to\theta_0$ almost surely.  Since $\theta_0$ is an interior
point of the convex stable neighborhood $\mathcal N$, the first-order
condition holds with probability tending to one.  The integral Taylor
expansion gives
\[
\begin{split}
    0
    ={}&
    \nabla_\theta L_T(\theta_0)\\
    &+
    \left[
      \int_0^1
      \nabla_\theta^2L_T
      \{\theta_0+x(\widehat\theta_T-\theta_0)\}
      \,dx
    \right]
    (\widehat\theta_T-\theta_0).
\end{split}
\]
Lemma~\ref{S-lem:hessian-ulln},
\eqref{S-eq:population-hessian}, and consistency imply that the matrix in
brackets converges in probability to $-\mathcal H$.  Lemma~\ref{S-lem:score-clt}
gives
\[
    \sqrt T\,\nabla_\theta L_T(\theta_0)
    =
    \frac1{\sqrt T}\sum_{t=1}^TS_t
    \Longrightarrow N(0,\Omega).
\]
Assumption~5 and Slutsky's theorem therefore yield
\eqref{S-eq:asymptotic-normality}.
\end{proof}

\subsubsection*{Initialization and sandwich inference}

Write $\widehat\ell_t(\theta)$ for the criterion evaluated from initialized
solutions of the score update equation, with the same known reference
distribution and Fisher projection matrices as before. Define
\[
 S_t(\theta)=\nabla_\theta\ell_t(\theta),\qquad
 \widehat S_t(\theta)=\nabla_\theta\widehat\ell_t(\theta),
\]
so that $S_t(\theta_0)=S_t$ in \eqref{S-eq:parameter-score}.
The score bounds, the Hessian formula \eqref{S-eq:likelihood-hessian}, and
the forgetting conditions in Assumptions~2 and 4 imply, for all sufficiently
large $t$,
\begin{equation}
\begin{split}
 \sup_{\theta\in\mathcal N}
 \bigl\{&\|\widehat S_t(\theta)-S_t(\theta)\|\\
       &+\|\nabla_\theta^2\widehat\ell_t(\theta)
                      -\nabla_\theta^2\ell_t(\theta)\|\bigr\}
 \le C_*\rho_*^t M_t,
 \qquad
 M_t=1+\sup_{\theta\in\mathcal N}
        \{\|D_t(\theta)\|^2+\|R_t(\theta)\|_F\},
\end{split}
\label{S-eq:initialization-derivative-bound}
\end{equation}
where $C_*<\infty$ almost surely and $\rho_*\in(0,1)$ is deterministic.
Here the bounded and Lipschitz spline score controls the score difference.
The bounded score Jacobian and its Lipschitz bound control the Hessian
difference. The remaining factors are bounded by $M_t$.

Assumption~4 makes $M_t$ stationary with finite expectation. In particular,
\[
 E\left[\sum_{t=1}^\infty\rho_*^t M_t\right]
 =\frac{\rho_*}{1-\rho_*}E[M_1]<\infty,
\]
so the weighted series is finite almost surely. The finitely many initial
terms are also finite, by continuity on the compact parameter neighborhood.
Thus both derivative differences in
\eqref{S-eq:initialization-derivative-bound} are summable uniformly in
$\theta$, almost surely. Their partial sums are therefore $O_p(1)$.
After the $T^{-1/2}$ score and $T^{-1}$ Hessian normalizations, these
differences vanish. Consistency of the initialized estimator and the same
Taylor argument then give the limit in \eqref{S-eq:asymptotic-normality}.
In the remainder, $\widehat\theta_T$ denotes the estimator computed from
these initial values.

\begin{remark}
Under Assumption~1, $\{S_t\}$
is a martingale difference sequence.  Its covariance limit is the
conditional-variance limit \eqref{S-eq:predictable-variance-limit}. The score
process itself need not be stationary.  The uniform $(2+\delta)$ moment
bound and a martingale strong law imply
\[
    \frac1T\sum_{t=1}^T
    \{S_tS_t'-E[S_tS_t'\mid\mathcal F_{t-1}]\}
    \longrightarrow0\quad\text{a.s.}
\]
For each fixed lag $k\ge1$, $S_tS_{t-k}'$ is also a martingale difference
with uniformly bounded moments of some order greater than one, so its
sample average converges to zero almost surely.  Thus the outer-product
estimator and a fixed-lag HAC estimator, including HAC(22)
\citep{NeweyWest1987}, converge in probability to $\Omega$.
This equivalence concerns parameter scores under Assumption~1. The
empirical parameter tests use lag zero; the HAC(22) calculations for
OOS forecast-score differences concern a different sequence and are
not justified by this martingale argument.
The continuity and moment conditions in Assumption~4 justify evaluating
these estimators at $\widehat\theta_T$. The forgetting conditions handle
initialization. In the special case
where the score and its conditional covariance form a jointly stationary
and ergodic process, the limit simplifies to $\Omega=E[S_tS_t']$.
The sandwich form
$\mathcal H^{-1}\Omega\mathcal H^{-1}$ remains necessary because
Assumption~1 allows both heterogeneity and dependence within a cross section,
so in general $\Omega\ne\mathcal H$.  More explicitly, with
$\psi_{t,i}=\psi\{u_{t,i},q_t(\theta_0)\}$,
\[
    \Sigma_t
    =\sum_{i=1}^{N_t}\sum_{j=1}^{N_t}w_{t,i}w_{t,j}
      \operatorname{Cov}(\psi_{t,i},\psi_{t,j}\mid\mathcal F_{t-1}).
\]
Under identical conditionally independent marginals,
\[
 \Sigma_t=(\sum_i w_{t,i}^2)\mathcal I_q\{q_t(\theta_0)\}.
\]
For equal weights, this reduces to $\mathcal I_q\{q_t(\theta_0)\}/N_t$.
Under heterogeneous marginals, conditional independence alone does not
give this simplification, because individual score covariances may differ
and conditional score means may be nonzero.  The variance estimator must use the outer
products of the date-level parameter scores, using lag zero in the main
implementation, rather than treating constituent scores as independent
replications. Wald tests also require a nonsingular asymptotic covariance
for the restrictions being tested.
\end{remark}

\begin{corollary}[Sandwich covariance and Wald inference]
\label{S-cor:oracle-sandwich-wald}
\leavevmode\par\noindent
Suppose that Assumptions~1--5 hold. At the estimator computed from fixed initial values,
let $\widehat S_t=\widehat S_t(\widehat\theta_T)$ and
$\bar S_T=T^{-1}\sum_t\widehat S_t$. Define the time-average matrices
\begin{align}
 \widehat{\mathcal H}_T
   &=-\frac1T\sum_{t=1}^T
               \nabla_\theta^2\widehat\ell_t(\widehat\theta_T),
       \label{S-eq:oracle-sensitivity-estimate}\\
 \widehat\Omega_T
   &=\frac1T\sum_{t=1}^T
       (\widehat S_t-\bar S_T)(\widehat S_t-\bar S_T)'.
       \label{S-eq:oracle-meat-estimate}
\end{align}
Then $\widehat{\mathcal H}_T\to\mathcal H$ and
$\widehat\Omega_T\to\Omega$ in probability. Consequently,
\begin{equation}
 \widehat V_T
   =\frac1T\widehat{\mathcal H}_T^{-1}
              \widehat\Omega_T\widehat{\mathcal H}_T^{-\prime}
   =H_T^{-1}B_TH_T^{-\prime},\qquad
 H_T=T\widehat{\mathcal H}_T,\quad B_T=T\widehat\Omega_T,
 \label{S-eq:oracle-sandwich-estimate}
\end{equation}
satisfies
\[
 T\widehat V_T\longrightarrow
 \mathcal V:=\mathcal H^{-1}\Omega\mathcal H^{-1}
 \quad\text{in probability}.
\]
This agrees with \eqref{eq:estimated-sandwich}.

For a fixed full-row-rank restriction matrix $R$ with $R\theta_0=r$ and
$R\mathcal V R'$ positive definite,
\begin{equation}
 W_T=(R\widehat\theta_T-r)'
          (R\widehat V_TR')^{-1}(R\widehat\theta_T-r)
       \Longrightarrow\chi^2_{\operatorname{rank}(R)}.
 \label{S-eq:oracle-wald-limit}
\end{equation}
The even and odd block tests each have three restrictions. Their joint
test has six. Positive definiteness is required for the covariance of the
tested restrictions, without requiring $\Omega$ to be positive definite
on every parameter direction.
\end{corollary}

\begin{proof}
The Hessian uniform law, consistency, and
\eqref{S-eq:initialization-derivative-bound} give
$\widehat{\mathcal H}_T\to\mathcal H$ in probability.
For the meat matrix at $\theta_0$, the martingale strong law in the preceding
remark and Assumption~5 yield
\[
 \frac1T\sum_t S_tS_t'\longrightarrow\Omega
 \quad\text{in probability},\qquad
 \frac1T\sum_tS_t\longrightarrow0\quad\text{a.s.}
\]
To justify substitution of the estimator, consider the predictable
one-observation parameter-score map
$J_t(\theta)'\psi\{u,q_t(\theta)\}$. Its modulus of continuity is uniform
in $u$, by the score bounds and Assumption~4. The corresponding modulus
for outer products of weighted cross-sectional scores is bounded by a constant times
$\sup_{\theta\in\mathcal N}\|D_t(\theta)\|^2$. This is an integrable,
stationary bound. Dominated convergence, the ergodic theorem for the
predictable modulus, and consistency therefore make the time-average
change in score outer products $o_p(1)$. The same argument applies to the
score mean. This step does not require stationarity of the observations.

Initialization also has a summable effect on the outer products. Indeed,
the bounded and Lipschitz score and the forgetting conditions give
\[
 \sup_{\theta\in\mathcal N}
 \|\widehat S_t(\theta)\widehat S_t(\theta)'
                  -S_t(\theta)S_t(\theta)'\|
 \le C_{**}\rho_*^t
       \left(1+\sup_{\theta\in\mathcal N}\|D_t(\theta)\|^2\right)
\]
for all sufficiently large $t$, with $C_{**}<\infty$ almost surely.
The summability argument above applies again. Hence
$\bar S_T=o_p(1)$, and subtracting $\bar S_T\bar S_T'$ does not affect
the meat limit. This proves $\widehat\Omega_T\to\Omega$.
Positive definiteness of $\mathcal H$ and continuity of inversion give
the sandwich limit. Theorem~\ref{S-thm:asymptotic-normality} for the estimator
computed from fixed initial values and Slutsky's theorem then give
\eqref{S-eq:oracle-wald-limit} under the stated restriction null.
\end{proof}

% Source: supp_simulation.tex
\section{Simulation calibration and additional forecasts}
\label{S-supp:simulation}

\subsection{Generating coefficients}
\label{S-supp:simulation-calibration}

The S\&P 500 HAR coefficient vectors used in generation, in horizon order
$(1,5,22)$, are
\begin{align*}
 \kappa_\eta^{\rm SP}&=(-5.745736,-4.815923,30.255739)',\\
 \kappa_{\rm even}^{\rm SP}&=(-0.074906,10.692708,259.195621)',\\
 \kappa_{\rm odd}^{\rm SP}&=(1.976005,-4.297596,63.792928)'.
\end{align*}
Each active vector is multiplied by $\sqrt{N/500}$ in both the independent
and ten-percent common-factor designs. Inactive vectors are zero.
The generating updates are given in
\eqref{eq:simulation-eta} and \eqref{eq:simulation-endpoint}.
The full-precision coefficients and replication-level results are recorded
in the companion repository.

\subsection{M5 compared with M3}
\label{S-supp:simulation-oos}

Table~\ref{S-tab:sim-m5-m3-oos} compares M5 with M3 after fitting M1--M5 in
each replication. This comparison evaluates the forecast gains from
adding tail-intensity dynamics to the fitted dynamic-$t$ reference, as in
the empirical analysis. Percentage improvements use the definitions in
Section~\ref{subsec:simulation-oos} of the main paper, with M3 replacing
M2 as the benchmark.

Both mean and median improvements in log score and one-percent quantile
loss are positive in all 18 cells of cases 2--4. Cases 0 and 1 show the
results when tail-intensity dynamics are absent from generation.
All entries use 100 paired replications with 642 OOS dates.

% Generated by build_sec4_sec5_tables.py from audited CSVs; do not edit numbers.
\begin{table}[htbp]
\centering
\caption{OOS improvements of M5 relative to M3: mean / median}
\label{S-tab:sim-m5-m3-oos}
\fontsize{9pt}{10.8pt}\selectfont
\setlength{\tabcolsep}{4pt}
\begin{tabular}{llrrr}
\hline
Case & Dependence & $N$ & Log score (\%) & 1\% quantile loss (\%)\\
\hline
0: Static & Independent & 20 & -0.010 / -0.008 & -0.084 / -0.077\\
 & Common 10\% & 20 & -0.012 / -0.010 & -0.113 / -0.089\\
 & Independent & 50 & -0.005 / -0.004 & -0.045 / -0.034\\
 & Common 10\% & 50 & -0.006 / -0.005 & -0.059 / -0.047\\
 & Independent & 100 & -0.002 / -0.001 & -0.016 / -0.017\\
 & Common 10\% & 100 & -0.005 / -0.004 & -0.042 / -0.033\\
\hline
1: $\eta$ only & Independent & 20 & -0.007 / -0.006 & -0.048 / -0.013\\
 & Common 10\% & 20 & -0.012 / -0.008 & -0.081 / -0.059\\
 & Independent & 50 & -0.004 / -0.003 & -0.030 / -0.032\\
 & Common 10\% & 50 & -0.007 / -0.004 & -0.069 / -0.037\\
 & Independent & 100 & -0.001 / -0.001 & -0.011 / -0.011\\
 & Common 10\% & 100 & -0.004 / -0.003 & -0.044 / -0.039\\
\hline
2: Even only & Independent & 20 & +0.038 / +0.026 & +0.396 / +0.271\\
 & Common 10\% & 20 & +0.038 / +0.029 & +0.397 / +0.370\\
 & Independent & 50 & +0.032 / +0.030 & +0.300 / +0.276\\
 & Common 10\% & 50 & +0.034 / +0.030 & +0.367 / +0.332\\
 & Independent & 100 & +0.026 / +0.023 & +0.251 / +0.275\\
 & Common 10\% & 100 & +0.041 / +0.029 & +0.417 / +0.300\\
\hline
3: $\eta$ and even & Independent & 20 & +0.035 / +0.029 & +0.343 / +0.323\\
 & Common 10\% & 20 & +0.023 / +0.023 & +0.169 / +0.105\\
 & Independent & 50 & +0.027 / +0.024 & +0.292 / +0.229\\
 & Common 10\% & 50 & +0.034 / +0.029 & +0.349 / +0.301\\
 & Independent & 100 & +0.029 / +0.026 & +0.279 / +0.247\\
 & Common 10\% & 100 & +0.042 / +0.035 & +0.402 / +0.350\\
\hline
4: $\eta$, even, and odd & Independent & 20 & +0.041 / +0.033 & +0.361 / +0.246\\
 & Common 10\% & 20 & +0.041 / +0.029 & +0.408 / +0.344\\
 & Independent & 50 & +0.034 / +0.028 & +0.349 / +0.312\\
 & Common 10\% & 50 & +0.044 / +0.040 & +0.434 / +0.338\\
 & Independent & 100 & +0.034 / +0.033 & +0.305 / +0.273\\
 & Common 10\% & 100 & +0.064 / +0.046 & +0.510 / +0.342\\
\hline
\end{tabular}
\par\smallskip
\begin{minipage}{\textwidth}
\fontsize{9pt}{10.8pt}\selectfont Percentage improvements over 100 paired replications with M1--M5 estimated. Each replication has 642 OOS dates. Positive values improve on M3. Means use the exponential of the mean log-score difference or the reduction in paired mean losses. Medians use replication-specific percentage improvements.
\end{minipage}
\end{table}

% Source: supp_empirical.tex
\section{Empirical data and PCA construction}
\label{S-supp:empirical}

\subsection{Constituent histories and observed returns}
\label{S-supp:data-construction}

Daily prices are obtained from Yahoo Finance. We use adjusted closing
prices when available and otherwise reported closing prices, and calculate
daily log returns. The S\&P 500 membership histories use
\texttt{fja05680/sp500}, taking the most recent snapshot dated on or before
each trading date. The final snapshot, June 30, 2026, is carried forward
to July 24, 2026. Nikkei 225 histories are reconstructed from the official
constituent set as of August 31, 2026 and the official component-change
ledger, replayed over the sample.

US sector assignments use public GICS Git snapshots and an audited
backcast for historical constituents. Japanese sectors combine official
current classifications with audited classifications for former
constituents. On each date, only constituents with observed returns and
estimable training PC1 loadings enter the raw and residual samples.
Table~\ref{S-tab:empirical-samples} reports the daily sample-size ranges
and training and OOS date counts for all 19 panels.

% Generated by build_sec4_sec5_tables.py from audited CSVs; do not edit numbers.
\begin{table}[htbp]
\centering
\caption{The 19 constituent-return cross sections}
\label{S-tab:empirical-samples}
\fontsize{10pt}{12pt}\selectfont
\setlength{\tabcolsep}{4pt}
\begin{tabular}{lrrrr}
\hline
Panel & $N_t$ range & Median $N_t$ & Training dates & OOS dates\\
\hline
S\&P 500 & 403--482 & 455 & 1886 & 642\\
Nikkei 225 & 212--223 & 218 & 1829 & 624\\
US Energy & 17--23 & 19 & 1886 & 642\\
US Materials & 19--27 & 25 & 1886 & 642\\
US Industrials & 54--77 & 70 & 1886 & 642\\
US Consumer Discretionary & 43--69 & 59 & 1886 & 642\\
US Consumer Staples & 30--36 & 32 & 1886 & 642\\
US Health Care & 41--63 & 56 & 1886 & 642\\
US Financials & 55--79 & 58 & 1886 & 642\\
US Information Technology & 52--72 & 61 & 1886 & 642\\
US Communication Services & 2--21 & 18 & 1886 & 642\\
US Utilities & 27--30 & 29 & 1886 & 642\\
US Real Estate & 26--30 & 30 & 1844 & 642\\
JP Technology & 52--59 & 56 & 1829 & 624\\
JP Financials & 18--21 & 20 & 1829 & 624\\
JP Consumer Goods & 27--36 & 33 & 1829 & 624\\
JP Materials & 48--60 & 57 & 1829 & 624\\
JP Capital Goods/Others & 35--36 & 36 & 1829 & 624\\
JP Transportation and Utilities & 19--21 & 19 & 1829 & 624\\
\hline
\end{tabular}
\par\smallskip
\begin{minipage}{\textwidth}
\fontsize{9pt}{10.8pt}\selectfont US and JP denote S\&P 500 and Nikkei 225 sectors. Raw and PC1 residual use the same constituent--dates, restricted to constituents with estimable training PC1 loadings. The broad-market panels contain constituent returns, rather than the return on an index. The Real Estate panel starts on September 1, 2016; the other panels start on July 5, 2016.
\end{minipage}
\end{table}

\subsection{Training-fixed PCA}
\label{S-supp:pca-details}

For each panel, we estimate covariance PCA from the training returns
through December 29, 2023. Missing training returns are replaced by their
constituent's observed training mean for this calculation only. The
centered matrix therefore has zero entries at missing observations.
Returns are not standardized by volatility. Constituents with no observed
training returns or no positive centered training variation are excluded.
The leading eigenvector has unit norm, and its sign is chosen to give a
nonnegative sum of loadings. Means and loadings remain fixed in training
and OOS evaluation.

The daily factor score and residuals are computed from observed returns
using the formula in Section~\ref{subsec:pca-robustness} of the main
paper. A date must have at least two eligible observed constituents and a
squared loading norm of at least $10^{-12}$. No imputed return enters
density estimation or forecast evaluation. Raw returns and PC1 residuals
use the same constituent--dates. The reported pairwise correlations use
matched training observations, and the PC1 shares are the leading
eigenvalues divided by the training covariance traces.